\documentclass[runningheads]{llncs}

\usepackage[T1]{fontenc}
\usepackage{graphicx}
\usepackage{enumerate}
\usepackage{amsmath}
\usepackage{algorithm}
\usepackage{algpseudocode}

\newtheorem{observation}[theorem]{Observation}

\begin{document}
\title{Adding an Edge in a $P_4$-sparse Graph
%\thanks{Research at the University of Ioannina supported by the Hellenic Foundation for Research and Innovation (H.F.R.I.) under the ``First Call for H.F.R.I. Research Projects to support Faculty members and Researchers and the procurement of high-cost research equipment grant'', Project FANTA (eFficient Algorithms for NeTwork Analysis), number HFRI-FM17-431.}
}
%
%\titlerunning{Abbreviated paper title}
% If the paper title is too long for the running head, you can set
% an abbreviated paper title here
%
\author{Anna Mpanti%\inst{1,2}
% \orcidID{0000-1111-2222-3333}
\and
Stavros D. Nikolopoulos%\inst{1,3}
% \orcidID{1111-2222-3333-4444}
\and
Leonidas Palios \thanks{Professor Leonidas Palios (1965–2026), esteemed colleague, co-author, and friend, passed away on January 21, 2026. His valuable contribution to our research and his lasting impact on our academic community will be remembered with deep respect.}%\inst{1,4}
%\orcidID{2222--3333-4444-5555}
}
\authorrunning{A. Mpanti et al.}
% First names are abbreviated in the running head.
% If there are more than two authors, 'et al.' is used.
%
\institute{Dept. of Computer Science and Engineering, University of Ioannina, Greece
% \and
% %0000-1111-2222-3333; \ 
% ampanti@cs.uoi.gr
% \and
% %0000-0001-6684-8459; \
% stavros@cs.uoi.gr
% \and
% %0000-0001-8630-3835; \ 
% palios@cs.uoi.gr
\\
\email{a.mpanti@uoi.gr} \quad
\email{sdnikolopoulos@uoi.gr} \quad
\email{palios@uoi.gr}
}
\maketitle              % typeset the header of the contribution

\begin{abstract}
For a graph family~${\cal F}$, the minimum ${\cal F}$-completion problem asks for the minimum number of non-edges to be added to a given graph~$G$ so that the resulting graph belongs to the family~${\cal F}$. This problem is NP-complete for many subclasses of perfect graphs, while polynomial solutions are available for minimal completion sets or restricted input.

We study the minimum $P_4$-sparse completion problem that arises after adding a prescribed non-edge to a \(P_4\)-sparse graph \(G\), with the resulting graph required to remain \(P_4\)-sparse. For the case in which $G$ consists of two connected components, each containing one endpoint of the added non-edge, we characterize the possible forms of an optimal solution. Combining this characterization with the solution of the other cases identified by the structural decomposition of \(P_4\)-sparse graphs, we obtain an algorithm for the general problem that runs in $O(n^2)$ time and $O(n)$ space.

\end{abstract}

\section{Introduction}
Given a graph~$G$, we wish to connect two non-adjacent vertices $u,v$ of $G$ by an edge; let us denote the resulting graph as $G + uv$. If $G$ belongs to a class~${\cal C}$ of graphs, this need not hold for the graph~$G + uv$. Hence, we are interested in computing the number of non-edges of $G$ in a minimum ${\cal C}$-completion of $G + uv$, i.e., the minimum number of non-edges (including the added non-edge~$u v$) to be added to $G + uv$ so that the resulting graph belongs to ${\cal C}$; such non-edges are called \emph{fill edges}. This is the (${\cal C},+1$)-Min\-Ed\-ge\-Ad\-di\-tion problem, a special case of the (${\cal C},+k$)-Min\-Ed\-ge\-Ad\-di\-tion Problem \cite{NP05} in which we add $k$ given non-edges in a graph belonging to a class~${\cal C}$ and we want to compute a minimum ${\cal C}$-completion of the resulting graph.

Computing a minimum completion of an arbitrary graph into a specific graph class is an important and well studied problem with applications in areas involving graph modeling with missing edges due to lacking data, e.g., molecular biology and numerical algebra \cite{GGKS95,NSS01,R72}. Unfortunately, minimum completions into many interesting graph classes, such as split graphs, chordal graphs and cographs, are NP-hard to compute \cite{BBD06,EC88,KF79,Ma08,Y81}. This led the researchers towards the computation of minimal completions \cite{CLPT17,HM09,HMP08,HP09,HTV05,ST09}, the solution of problems with restricted input \cite{BKKM98,BDK00,KKS97,KKW98,M05}, parameterized algorithms \cite{KST04,Ma07,VHPT09}, and approximation algorithms \cite{NSS98}.

The problem of adding an edge to a graph is motivated by the dynamic recognition (or on-line maintenance) problem on graphs: a series of requests for the addition or the deletion of an edge or a vertex (potentially incident on a number of edges) are submitted and each is executed only if the resulting graph remains in the same class of graphs. Several authors have studied this problem for different classes of graphs and have given algorithms supporting some or all the above operations; we mention the edges-only fully dynamic algorithm of Ibarra \cite{I08} for chordal and split graphs, and the fully dynamic algorithms of Hell et al. \cite{HSS02} for proper interval graphs, of Shamir and Sharan \cite{SS04} for cographs, and of Nikolopoulos et al{.} for $P_4$-sparse graphs \cite{NPP12}.

In this paper, we consider the ($P_4$-spar\-se,$+1$)-Min\-Ed\-ge\-Ad\-di\-tion problem; this is a generalization of the (co\-graph,$+1$)-Min\-Ed\-ge\-Ad\-di\-tion problem studied by Nikolopoulos and Palios \cite{NP05} who presented an algorithm that runs in time linear in the number of vertices of the given cograph.
For the problem on $P_4$-sparse graphs, we prove that for a
$P_4$-sparse graph $G$ consisting of two connected components each
containing one of the endpoints of the added non-edge $uv$, we
characterize the possible forms of an optimal solution of the
$(P_4\text{-sparse},+1)$-MinEdgeAddition problem, including the case
in which the prescribed edge becomes a leg of a thin spider.
This, along with the solution of the problem when the added non-edge
connects two non-adjacent vertices of a spider, enables us to present
an algorithm for the problem which requires time quadratic in the
number of vertices of the given $P_4$-sparse graph.

\begin{figure}[t]
\centering
\includegraphics[width=3.3in]{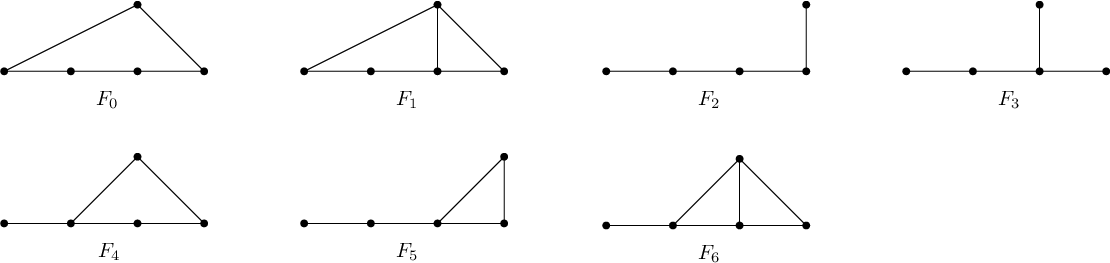}
\caption{The forbidden subgraphs of the class of $P_4$-sparse graphs.}\label{fig:forbidden}
\end{figure}

\section{Theoretical Framework}
%
%We consider finite undirected graphs with no loops or multiple edges. For a graph~$G$, we denote by $V(G)$ and $E(G)$ the vertex set and edge set of $G$, respectively. Let $S$ be a subset of the vertex set~$V(G)$ of a graph~$G$. Then, the subgraph of $G$ induced by $S$ is denoted by $G[S]$.
%% Moreover, we denote by $G - S$ the subgraph $G[V(G) - S]$. The \emph{union} $G_1 \cup G_2$ of graphs $G_1, G_2$ with disjoint vertex sets $V(G_1)$ and $V(G_2)$ is the graph with vertex set~$V(G_1) \cup V(G_2)$ and edge set $E(G_1) \cup E(G_2)$. The \emph{join} $G_1+G_2$ of graphs $G_1, G_2$ with disjoint vertex sets $V(G_1)$ and $V(G_2)$is the graph union $G_1 \cup G_2$ together with all the edges joining $V(G_1)$ and $V(G_2)$.

All graphs are finite, undirected, and simple. For a graph \(G\), \(V(G)\) and \(E(G)\) denote its vertex and edge sets, and for \(S\subseteq V(G)\), \(G[S]\) denotes the induced subgraph. If \(G_1,G_2\) have disjoint vertex sets, then \(G_1\cup G_2\) is their union and \(G_1+G_2\) is their join, obtained by adding all edges between \(V(G_1)\) and \(V(G_2)\). For \(x\in V(G)\), \(N_G(x)\) and \(N_G[x]=N_G(x)\cup\{x\}\) denote its open and closed neighborhoods. For \(S\subseteq V(G)\),
$N_G(S)=\left(\bigcup_{x\in S}N_G(x)\right)\setminus S$, $N_G[S]=N_G(S)\cup S$. The degree of \(x\) is $\deg_G(x)=|N_G(x)|
$. A vertex is universal if it is adjacent to all other vertices of the graph; it is universal in a set \(S\subseteq V(G)\) if it is universal in the induced subgraph \(G[S]\). The chordless path on \(k\) vertices is denoted by \(P_k\).

Finally, we note that for any two positive integers $i_1, i_2$, it holds that

 \begin{equation} \label{eq:product}
 \begin{split}
 & i_1 \cdot i_2 \ \ge\  i_1 + i_2 - 1;
 \hbox{equality holds if } i_1 = 1 \hbox{ or } i_2 = 1.
 \end{split}
 \end{equation}
The equality condition follows from: 
%the fact that 
$i_1 \cdot i_2 = i_1 + i_2 - 1 \ \Longleftrightarrow\ (i_1 - 1) \cdot (i_2 - 1) = 0$.

\subsubsection{$P_4$-sparse Graphs.}
%\medskip\noindent
%\textbf{$P_4$-sparse Graphs}. 
\ A graph in which every set of $5$ vertices induces at most one $P_4$ is \emph{$P_4$-sparse} \cite{H85} (Figure~\ref{fig:forbidden} depicts the $7$ forbidden subgraphs for the class of $P_4$-sparse graphs \cite{JO92b}). The $P_4$-sparse graphs are perfect and also perfectly orderable \cite{H85}, and properly contain many graph classes, such as, the cographs, the $P_4$-reducible graphs, etc (see \cite{BLS99,JO92a,JO92b}). They have received considerable attention in recent years and find applications in applied mathematics and computer science (e.g., communications, transportation, clustering, scheduling, computational semantics) in problems that deal with graphs 
%featuring “local density” properties.
exhibiting ``local density''.
%; note that the notion of local density is often associated with the absence of $P_4$s.

For a $P_4$-sparse graph either the graph or its  complement is disconnected with the connected components inducing $P_4$-sparse graphs, or the graph induces a spider. 
A graph $H$ is called a \emph{spider} if its vertex set~$V(H)$ admits a partition into sets $S, K, R$ such that:
\begin{itemize}
\item
the set~$S$ is an independent
%(stable) 
set, the set~$K$ is a clique, and $|S| = |K| \ge 2$;
\item
every vertex in $R$ is adjacent to every vertex in $K$ and to no vertex in $S$;
\item
there exists a bijection $f: S \to K$ such that either $N_G(s) \cap K = \{f(s)\}$ for each vertex~$s \in S$ or else, $N_G(s) \cap K = K \setminus \{f(s)\}$ for each vertex~$s \in S$; in the former case, the spider is \emph{thin}, in the latter it is \emph{thick} (Figure~\ref{fig:spiders}). In a thin spider with spider partition $(S,K,R)$ and matching bijection $f:S\to K$, each edge $sf(s)$, for $s\in S$, is called a leg of the spider. The vertices $s$ and $f(s)$ are called matched vertices.
\end{itemize}
The triple $(S, K, R)$ is called the \emph{spider partition}. If $R = \emptyset$ then the spider is \emph{headless}. Note that for $|S| = |K| = 2$, the spider is 
%simultaneously 
thin and thick; thus, to avoid ambiguity, in the following, we assume that $|K| \ge 3$ in thick spiders. 
\begin{figure}[t!]
\centering
%\hrule\medskip\medskip	
\includegraphics[width=2in]{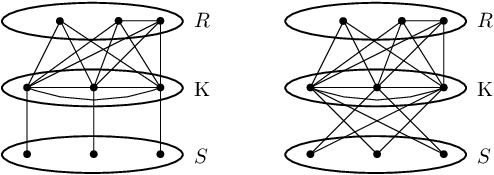}
%\smallskip\medskip\hrule\medskip
\caption{(left)~A thin spider; (right)~a thick spider.}
\label{fig:spiders}
\end{figure}
In \cite{JO92b}, Jamison and Olariu showed that each $P_4$-sparse graph~$G$ admits a unique tree representation, up to isomorphism, called the \emph{$P_4$-sparse tree} $T_G$ of $G$ which is a rooted tree such that:
\begin{itemize}
\item[(i)]
each internal node of $T_G$ has at least $2$ children provided that $|V(G)| \ge 2$;
\item[(ii)]
the internal nodes of $T_G$ are labelled by one of $0$, $1$, $2$ (\emph{$0$-, $1$-, $2$-nodes, resp.}) and the parent-node of each $0$- or $1$-node~$t$ has a different label than $t$;
\item[(iii)]
the leaves of $T_G$ are in a $1$-to-$1$ correspondence with the vertices of $G$; if the least common ancestor of the leaves corresponding to two vertices $v_i, v_j$ of $G$ is a $0$-node ($1$-node, resp.) then the vertices $v_i, v_j$ are non-adjacent (adjacent, resp.) in $G$, whereas the vertices corresponding to the leaves of a subtree rooted at a $2$-node induce a spider.
\end{itemize}

%The structure of the $P_4$-sparse tree implies the following lemma.

%\begin{lemma} \label{lemma:p4sparse_tree}
%Let $G$ be a $P_4$-sparse graph and let $H = (S,K,R)$ be a thin spider of $G$. Moreover, let $s \in S$ and $k \in K$ be vertices that are adjacent in the spider.
%\par\smallskip\noindent
%\textbf{\ P1}.
%Every vertex of the spider is adjacent to all vertices in $N_G(s) - \{k\}$.
%\par\smallskip\noindent
%\textbf{\ P2}.
%Every vertex in $K - \{k\}$ is adjacent to all vertices in $N_G(k) - \{s\}$.
%\end{lemma}
%\par\smallskip

%\medskip\noindent
%\textbf{Note}. With a slight abuse of terminology, in the following, we will simply use the term \emph{edges} instead of fill edges, which in fact are non-edges of the given graph.

\section{Adding a Non-edge incident on a Vertex of the Clique or the Independent Set of a Spider}
\label{ch:inside_spider}

In this section, we consider the ($P_4$-sparse,{$+$}$1$)-Min\-Edge\-Addition Problem for a spider $G = (S,K,R)$ and a non-edge~$e$ incident on a vertex in $S \cup K$.
In the following, for simplicity, we assume that $S = \{ s_1, s_2, \ldots, s_{|K|} \}$, $K = \{ k_1, k_2, \ldots, k_{|K|} \}$, and $R = \{ r_1, r_2, \ldots, r_{|R|} \}$ where $|K| \ge 2$ and $|R| \ge 0$.

%\subsection{Thin Spider}
\noindent
\textbf{\emph{Thin Spider.}} Suppose that the spider~$G$ is thin and that $s_i$ is adjacent to $k_i$ for each $i = 1, 2, \ldots, |K|$. The following lemmas address the cases of the addition of a non-edge~$e$.

\begin{lemma} \label{lemma:thin}
Consider the ($P_4$-sparse,{$+$}$1$)-Min\-Edge\-Addition Problem for the thin spider $G = (S,K,R)$ and a non-edge~$e$.
\begin{itemize}
\item[(i)]
If $e$ is incident on a vertex in $S$ and a vertex in $K$, the number of fill edges needed (including $e$) is $|K|-1$.
\item[(ii)]
If both endpoints of $e$ belong to $S$, the number~$N$ of fill edges needed (including $e$) is
    \[
	N =
	\begin{cases}
	\ 2\, |K| - 3, & \text{if $|R| = 0$;}\\
	\ 2\, |K| - 2, & \text{if $|R| = 1$;}\\
	\ 2\, |K| - 1, & \text{if $|R| \ge 2$.}
	\end{cases}
	\]
\item[(iii)]
If $e$ is incident on a vertex~$s$ in $S$ and a vertex in $R$, the number of fill edges needed (including $e$) is $|K| - 1 + \mu$ where $\mu$ is the number of fill edges in an optimal solution of the ($P_4$-sparse,{$+$}$1$)-Min\-Edge\-Addition Problem for the disconnected induced subgraph~$G[\{s\} \cup R]$ and the non-edge~$e$.
\end{itemize}
\end{lemma}

\begin{proof}
(i)~Suppose, without loss of generality, that $e = k_1 s_2$.
Then, we can get a $P_4$-sparse graph if, in addition to the fill edge~$e$, we add the fill edges $k_1 s_j$ ($j = 3, \ldots, |K|$) or alternatively the fill edges $s_2 k_j$ ($j = 3, \ldots, |K|$) for a total of $|K|-1$ fill edges.

For optimality, let $G'$ be any solution and set $X_j=\{s_1,k_1,s_2,s_j,k_j\}$ for $j=3,\ldots,|K|$. If no fill edge in $G'[X_j]$ is incident with $s_j$ or $k_j$, then $G'[X_j]$ induces $F_3$ or $F_5$, according as $s_1s_2$ is absent or present. Hence each $X_j$ requires a fill edge incident with $s_j$ or $k_j$. These edges are distinct since $X_i\cap X_j=\{s_1,k_1,s_2\}$ for $i\ne j$, giving at least $|K|-1$ fill edges in total.    
% To prove optimality, let $G'$ be any solution. For each
% $j=3,\ldots,|K|$, let $X_j=\{s_1,k_1,s_2,s_j,k_j\}$. Since $P_4$-sparseness is hereditary, $G'[X_j]$ must be
% $P_4$-sparse. If no fill edge in $G'[X_j]$ is incident with $s_j$ or $k_j$, then the only possible additional fill edge in $X_j$, besides $e$, is $s_1s_2$. Thus, $G'[X_j]$ induces $F_3$ if $s_1s_2$ is not added, and $F_5$ if it is added, a contradiction. Hence, for every $j=3,\ldots,|K|$, at least one fill edge in $X_j$ is incident with $s_j$ or $k_j$.

% Such fill edges are distinct for different values of $j$, since $X_i\cap X_j=\{s_1,k_1,s_2\}$ for $i\ne j$. Therefore at least $|K|-2$ additional fill edges are required besides $e$, and hence every solution contains at least $|K|-1$ fill edges.

\medskip\noindent
(ii)~Suppose, without loss of generality, that $e = s_1 s_2$.
Then, we can get a $P_4$-sparse graph if, in addition to the fill edge~$e$, we add the following fill edges:
\begin{itemize}
	\item
		if $R = \emptyset$, $s_1 k_3$, $\ldots$, $s_1 k_{|K|}$ and $s_2 k_3$, $\ldots$, $s_2 k_{|K|}$ in which case $s_1, s_2, k_1, k_2$ become adjacent to all vertices in $K \setminus \{k_1,k_2\}$;
	\item
		if $R = \{r_1\}$, $s_1 k_3$, $\ldots$, $s_1 k_{|K|}$, $s_2 k_3$, $\ldots$, $s_2 k_{|K|}$, and $1$ fill edge (among $r_1 s_1$, $r_1 s_2$, $k_1 s_2$, $k_2 s_1$) so that the  forbidden subgraph~$F_1$ induced by $s_1$, $s_2$, $k_1$, $k_2$, $r_1$ becomes a $P_4$-sparse graph;
	\item
		if $|R| \ge 2$, $k_1 s_2$, $k_1 s_3$, $\ldots$, $k_1 s_{|K|}$ and $k_2 s_1$, $k_2 s_3$, $\ldots$, $k_2 s_{|K|}$, (then $k_1, k_2$ become universal in $G$);
\end{itemize}
for a total of $N$ fill edges as stated above.
	
	\begin{figure}[t]
		\hrule\medskip\medskip	
		\centering
		\includegraphics[width=0.8\textwidth]{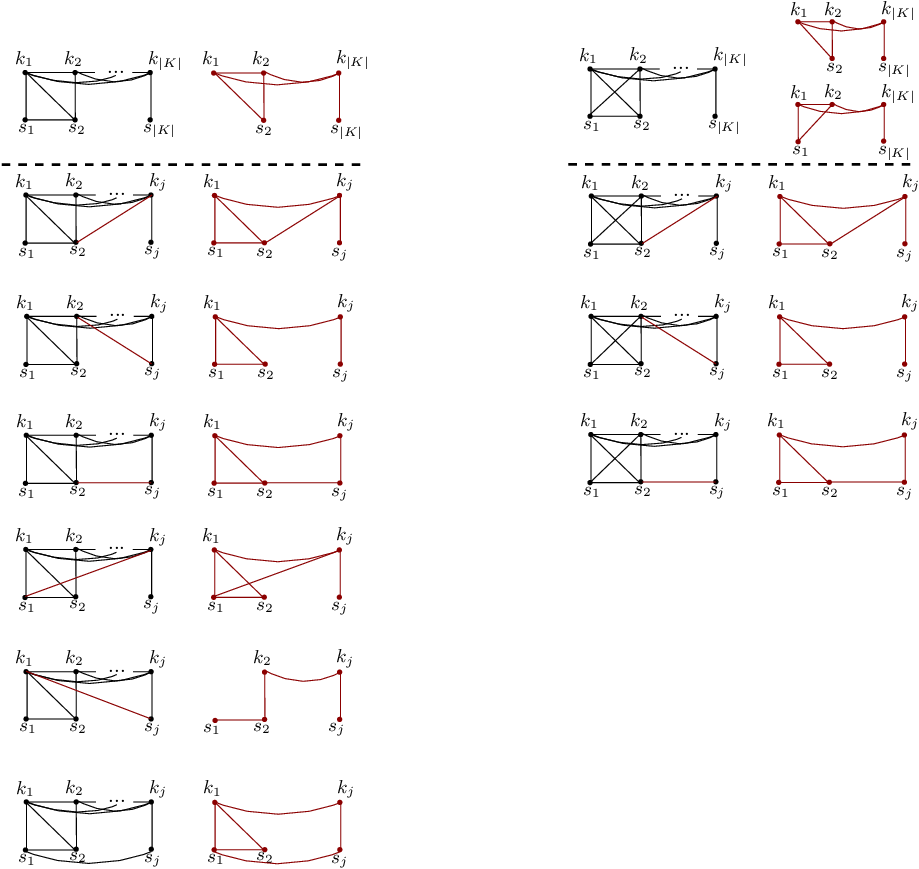}
		\smallskip\medskip\hrule\medskip
		\caption{Configurations used in the proof of Lemma~\ref{lemma:thin}(ii). Left: the fill edges $s_1s_2$ and $k_1s_2$ are added, but not $k_2s_1$; right: $s_1s_2$, $k_1s_2$, and $k_2s_1$ are added. The lower configurations show the addition of one further fill edge; in each case, the red graph is an induced forbidden subgraph.}
		\label{fig:intra_s-s}
	\end{figure}

For $|K|=2$, if $|R|=0$, fewer than $N=1$ fill edges is impossible since $e$ itself must be added. If $|R|=1$, $G+e$ induces $F_1$, so one additional fill edge is necessary. Suppose $|R|\ge2$. If neither $k_1s_2$ nor $k_2s_1$ is added, at least one fill edge incident with each $r\in R$ is required, giving at least $|R|+1\ge3=N$ fill edges. If exactly one is added, say $k_1s_2$, then $\{s_1,s_2,k_2,r_1,r_2\}$ induces $F_5$ or $F_3$, so a third fill edge is necessary; if both are added, there are already three fill edges including $e$. Thus fewer than $N$ fill edges is impossible.
    
% To prove the optimality of this solution, suppose for contradiction that there exists an optimal solution~$G'$ that requires \emph{fewer} than $N$ fill edges (including $e$).
% First, consider that $|K| = 2$. Then the values of $N$ imply that $G'$ requires at most $0$~fill edges if $|R| = 0$, at most $1$~fill edge if $|R| = 1$, and at most $2$~fill edges if $|R| \ge 2$ including $e$ in each case. The number of fill edges if $|R| = 0$ leads to a contradiction since $e$ is added. If $|R| = 1$, then the addition of $e$ results in an $F_1$ ($=$~house) and at least $1$ additional fill edge needs to be added, a contradiction again. 
% If $|R| \ge 2$, then each vertex~$r \in R$ and the vertices $s_1$, $s_2$, $k_1$, $k_2$ induce an $F_1$, and additional fill edges are needed. If neither the fill edge~$k_1 s_2$ nor the fill edge~$k_2 s_1$ is added, then we need $1$ fill edge incident on each $r \in R$; since $G'$ requires at most $2$ fill edges (including $e$) then $|R| + 1 \le 2 \Longleftrightarrow |R| \le 1$, in contradiction to the fact that $|R| \ge 2$. Since $G'$ uses at most $2$ fill edges including $e$, only one of $k_1 s_2$ and $k_2 s_1$ can be added; let that be the fill edge~$k_1 s_2$. But then the vertices $s_1$, $s_2$, $k_2$, $r_1$, $r_2$ induce a forbidden subgraph~$F_5$ or $F_3$ (depending on whether $r_1, r_2$ are adjacent or not), a contradiction.

Now, consider that $|K| \ge 3$.
	
\noindent
\textit{A. Suppose that neither the non-edge~$k_1 s_2$ nor the non-edge~$k_2 s_1$ is added.} Then, the vertices $k_1, k_2, s_1, s_2$ induce a $C_4$. For each vertex $k_j$ ($3 \le j \le |K|$), the vertices $k_1, k_2, s_1, s_2, k_j$ induce a forbidden subgraph~$F_1$ and thus for each subgraph at least one fill edge needs to be added; since $k_1 s_2$ and $k_2 s_1$ cannot be added, this fill edge has to be adjacent to $k_j$ (connecting it to $s_1$ or $s_2$). If only one of these two non-edges is added, say the fill edge~$k_j s_2$ but not the non-edge~$k_j s_1$, then a fill edge is needed adjacent to $s_j$, otherwise the vertices $k_1, s_1, s_2, k_j, s_j$ induce a forbidden subgraph~$F_4$. Thus, for each $j = 3, \ldots, |K|$, we need to add at least $2$ fill edges, for a total of $2 \, |K| - 4$ fill edges in addition to $e$. Moreover, if $|R| > 0$, for each vertex~$r_i \in R$, the vertices $k_1, k_2, s_1, s_2, r_i$ induce a forbidden subgraph~$F_1$ and thus at least $1$ additional fill edge adjacent to $r_i$ needs to be added. Then, the total number of fill edges is at least equal to $|R| + 2 \, |K| - 3$ (including $e$), which is no less than the value of $N$ for all values of $|R|$.
	
\smallskip\noindent
\textit{B. Suppose that exactly one of the non-edges $k_1 s_2$ and $k_2 s_1$ is added.} Without loss of generality, suppose that the fill edge~$k_1 s_2$ is added (and not the non-edge~$k_2 s_1$). Then, for each pair of vertices $s_j, k_j$ ($3 \le j \le |K|$), the vertices $k_1, k_2, s_2, k_j, s_j$ induce a forbidden subgraph~$F_6$. But a single fill edge is not enough (as shown in Figure~\ref{fig:intra_s-s}  (right)). Thus, the total number of fill edges (including $e$) is at least $3 + 2\, (|K|-2) = 2\,|K| - 1$, which is no less than the value of $N$ for $|R| \le 1$. If $|R| \ge 2$, the vertices $s_1, s_2, k_2, r_1, r_2$ induce a forbidden subgraph~$F_5$ or $F_3$ (depending on whether $r_1, r_2$ are adjacent or not); hence, at least one more fill edge is needed, for a total of $2 \, |K| - 1$ fill edges (including $e$), which is no less than the value of $N$ for $|R| \ge 2$.
	
\smallskip\noindent
\textit{C. Suppose that both the edges $k_1 s_2$ and $k_2 s_1$ are added.} Then, the vertices $s_1$, $s_2$, $k_1$, $k_2$ induce a $K_4$. For each pair of vertices $k_j, s_j$ ($3 \le j \le |K|$), the vertices $s_1$, $k_1$, $k_2$, $k_j$, $s_j$ and $k_1$, $k_2$, $s_2$, $k_j$, $s_j$ induce a forbidden subgraph~$F_5$. But a single fill edge is not enough (as shown in Figure~\ref{fig:intra_s-s} (right)). Thus, the total number of fill edges (including $e$) is at least $3 + 2\, (|K|-2) = 2\,|K| - 1$, which is no less than the value of $N$ for all values of $|R|$.

\medskip\noindent
(iii)~Suppose, without loss of generality, that $s = s_1$ and $e = s_1 r_1$ with $r_1 \in R$.
Then, we can get a $P_4$-sparse graph if first we add the fill edges $s_1 k_j$ ($j = 2,3, \ldots, |K|$) which makes $s_1$ adjacent to all the vertices in $K$ and then add the minimum number of fill edges so that the disconnected induced subgraph~$G[\{s_1\} \cup R]$ with the non-edge~$s_1 r_1$ becomes $P_4$-sparse for a total of $|K| - 1 + \mu$ fill edges (including $e$).

For optimality, let $G'$ be any solution and put $Y=\{s_1\}\cup R$. For each $j=2,\ldots,|K|$, the set $X_j=\{s_1,k_1,r_1,k_j,s_j\}$ induces $F_6$ after adding $e$, so at least one fill edge incident with $k_j$ or $s_j$ is needed. These are distinct and lie outside $Y$, giving $|K|-1$ fill edges outside $G'[Y]$. Since $G'[Y]$ must itself contain at least $\mu$ fill edges, the total is at least $|K|-1+\mu$.
% To prove optimality, let $G'$ be any solution and put
% $Y=\{s_1\}\cup R$. For each $j=2,\ldots,|K|$, let $X_j=\{s_1,k_1,r_1,k_j,s_j\}$.After the addition of $e=s_1r_1$, the vertices of $X_j$ induce the forbidden subgraph~$F_6$. Since $P_4$-sparseness is hereditary, $G'[X_j]$ must be $P_4$-sparse. Hence, for every $j=2,\ldots,|K|$, at least one fill edge in $X_j$ must be incident with $k_j$ or $s_j$.

% These fill edges are distinct for different values of $j$, and none of them is contained in $Y$. Therefore at least $|K|-1$ fill edges are required outside $G'[Y]$. Moreover, $G'[Y]$ is a $P_4$-sparse completion of $G[Y]$ containing the prescribed edge $s_1r_1$. Hence, by the definition of $\mu$, it contains at least $\mu$ fill edges. Thus every solution has at least $|K|-1+\mu$ fill edges.
\end{proof}

%\subsection{Thick Spider}
\noindent
\textbf{\emph{Thick Spider.}} Suppose that the spider~$G$ is thick and that $s_i$ is non-adjacent to $k_i$ for each $i = 1, 2, \ldots, |K|$. Additionally, according to our convention, we assume that $|K| \ge 3$.

\smallskip
\noindent
\begin{lemma} \label{lemma:thick}
Consider the ($P_4$-sparse,{$+$}$1$)-MinEdgeAddition Problem for the thick spider $G = (S,K,R)$ and a non-edge~$e$.
\begin{itemize}
\item[(i)]
If $e$ is incident on a vertex in $S$ and a vertex in $K$, no fill edges are needed in addition to $e$.
\item[(ii)]
If both endpoints of $e$ belong to $S$, the number~$N$ of fill edges needed (including $e$) is
	\[
	N =
	\begin{cases}
	%\ 1, & \text{if $|K| = 3$ and $R = \emptyset$;}\\
	\ 2, & \text{if $|K| = 3$ and $R = \emptyset$;}\\
	\ 3, & \text{if $|K| + |R| \ge 4$.}
	\end{cases}
	\]
\item[(iii)]
If $e$ is incident on a vertex~$s$ in $S$ and a vertex in $R$, the number of fill edges needed (including $e$) is $1 + \mu$ where $\mu$ is the number of fill edges in an optimal solution of the ($P_4$-sparse,{$+$}$1$)-Min\-Edge\-Addition Problem for the disconnected induced subgraph~$G[\{s\} \cup R]$ and the non-edge~$e$.
\end{itemize}
\end{lemma}
\begin{proof}
(i)~Suppose, without loss of generality, that $e = k_1 s_1$.
Then, the addition of $e$ makes $k_1$ universal, and no additional fill edges are needed, which is optimal.

\medskip\noindent
(ii)~Suppose, without loss of generality, that $e = s_1 s_2$. Additionally, recall that we assume that $|K| \ge 3$. We can get a $P_4$-sparse graph if, in addition to the fill edge~$e$, we add the fill edge~$s_2 s_3$ if $|K| = 3$ and $|R| = 0$ (note that the complement of the resulting graph is the union of the $P_2$~$s_2 k_2$ and the $P_4$~$k_1 s_1 s_3 k_3$) and the fill edges~$s_1 k_1$ and $s_2 k_2$ if $|K| + |R| \ge 4$ (note that $k_1, k_2$ are universal in the resulting graph).
	
To establish the optimality of this solution, we first observe that
% for $|K| = 2$, the number of fill edges matches the minimum number of fill edges for a thin spider with $|K| = 2$, and that 
for $|K| = 3$ and $|R| = 0$, the vertices $s_1, s_2, s_3, k_1, k_2$ induce a forbidden subgraph~$F_1$ and thus, at least $2$ fill edges (including $e$) are needed.
Next we show that for $|K| + |R| \ge 4$, no solution has fewer than $3$ fill edges (including $e$).
Suppose for contradiction that there is a solution with at most $2$ fill edges. Due to $e$, the vertices $s_1$, $s_2$, $s_3$, $k_1$, $k_2$ induce a forbidden subgraph~$F_1$, and thus at least $1$ additional fill edge is needed.
	
\smallskip\noindent
\textit{A. This additional fill edge is $s_1 k_1$ or $s_2 k_2$.} Due to symmetry, suppose without loss of generality that the fill edge~$s_1 k_1$ is added. But then, the vertices $s_1$, $s_2$, $s_3$, $k_2$, $k_3$ induce a forbidden subgraph~$F_6$, a contradiction.
	
\smallskip\noindent
\textit{B. None of the non-edges $s_1 k_1$ and $s_2 k_2$ is added.} Then, the vertices $s_1$, $s_2$, $k_1$, $k_2$ induce a $C_4$ and for each $q \in \{s_3, \ldots, s_{|K|} \} \cup R$, the vertices $s_1$, $s_2$, $k_1$, $k_2$, $q$ induce a forbidden subgraph~$F_1$ and either the fill edge~$q s_1$ or the fill edge~$q s_2$ needs to be added (recall that none of $s_1 k_1$, $s_2 k_2$ is added). Since for the different possibilities of $q$, these fill edges are distinct and at most $1$ fill edge is added in addition to $e$, then it must hold that $|K| + |R| - 2 = 1 \Longleftrightarrow |K| + |R| = 3$, in contradiction to the fact that $|K| + |R| \ge 4$.

\medskip\noindent
(iii)~Suppose, without loss of generality, that $s = s_1$ and $e = s_1 r_1$ with $r_1 \in R$.
Then, we can get a $P_4$-sparse graph if first we add the fill edge $s_1 k_1$ (this makes $k_1$ universal and $s_1$ adjacent to all the vertices in $K$), and then add the minimum number~$\mu$ of fill edges (including $e$) so that the disconnected induced subgraph~$G[\{s_1\} \cup R]$ with the non-edge~$e$ becomes $P_4$-sparse for a total of $1 + \mu$ fill edges.
	
The optimality of this solution follows from the fact that, due to the addition of the non-edge~$e = s_1 r_1$, the vertices $s_1$, $s_2$, $k_1$, $k_2$, $r_1$ induce a forbidden subgraph~$F_6$ and so at least $1$ fill edge incident on a vertex in $S \cup K$ and other than $e$ is needed.
\end{proof}

%============
\begin{figure}[t]
\centering

\begin{minipage}[c]{0.30\textwidth}
\centering
\includegraphics[width=\linewidth]{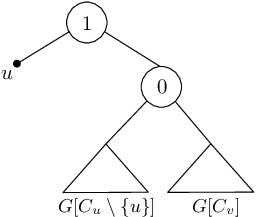}

\small (a)
\end{minipage}
\hfill
\begin{minipage}[c]{0.35\textwidth}
\centering
\includegraphics[width=0.88\linewidth]{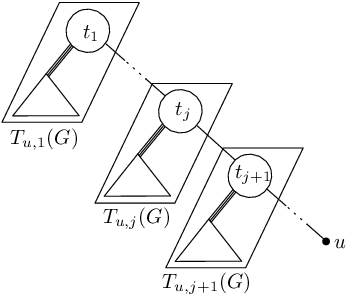}

\small (b)
\end{minipage}
\hfill
\begin{minipage}[c]{0.29\textwidth}
\centering
\includegraphics[width=0.99\linewidth]{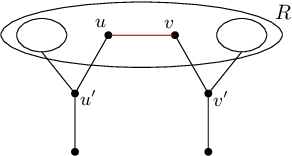}

\small (c)
\end{minipage}

\caption{(a) The tree representation $T_G$ with $u$ made universal in $G$.
(b) The subtrees $T_{u,1}(G),T_{u,2}(G),\ldots$ along the path to $u$.
(c) In $G[C_u]$, $u$ is adjacent only to $u'$, which is universal in
$G[C_u]$; symmetrically for $v$.}
\label{fig:trees}
\end{figure}
%===============

% \begin{figure}[t]
% %\hrule\medskip\medskip	
% \centering
% \includegraphics[width=0.34\linewidth]{u_universal.eps}
% %\smallskip\medskip\hrule\medskip
% \caption{The tree representation $T_G$ with the vertex~$u$ made universal in $G$.}
% \label{fig:u_universal}
% \end{figure}

\section{Connecting two Connected Components}
\label{sec:cc}

In this section, we will consider the special case in which the given $P_4$-sparse graph~$G$ consists of $2$ connected components each containing one of the endpoints of the added non-edge~$u v$; we will call this the \emph{($P_4$-sparse-$2$CC,$+1$)-Min\-Ed\-ge\-Ad\-di\-tion Problem}. Let $C_u$ ($C_v$ resp.) be the connected component of $G$ containing $u$ ($v$ resp.). Clearly $V(G) = C_u \cup C_v$. 
First, we note that we can get a $P_4$-sparse graph~$G'$ where $V(G') = V(G)$ and $E(G) \cup \{uv\} \subseteq E(G')$ by making $u$ universal in $G$ (Figure~\ref{fig:trees}(a)) which requires
\begin{equation} \label{eq:u_universal}
|V(G)| - 1 - deg_G(u)
\end{equation}
fill edges (including $uv$). 
%A similar statement holds for $v$.
Similarly, for $v$, the number of fill edges (including $uv$) is
\begin{equation} \label{eq:v_universal}
|V(G)| - 1 - deg_G(v).
\end{equation}

%\begin{corollary} \label{corol:universal}
%The ($P_4$-sparse,+$1$)-Min\-Edge\-Addition Problem for a $P_4$-sparse graph~$G$ and the non-edge $uv$ has a solution using fewer that $|V(G)|$ fill edges.
%\end{corollary}

Next, we prove a lemma which will be useful in subsequent lemma proofs.

\begin{lemma} \label{lemma:useful}
Let $H$ be an optimal solution of the
($P_4$-spar\-se-2CC,$+1$)-Min\-Ed\-ge\-Ad\-di\-tion Problem for a
$P_4$-sparse graph~$G$ and a non-edge~$uv$, and suppose that the root
node of the $P_4$-sparse tree~$T_H$ of $H$ is a $1$-node.
If $t$ is a child of the root of $T_H$ such that, for the set~$A$ of
vertices of $G$ corresponding to the leaves of the subtree rooted at
$t$, it holds that
$A\cap C_u\ne\emptyset$, $
A\cap C_v\ne\emptyset$, $
(V(G)\setminus A)\cap C_u\ne\emptyset$, $
(V(G)\setminus A)\cap C_v\ne\emptyset,$
then there exists an optimal solution of the
($P_4$-spar\-se-2CC,$+1$)-Min\-Ed\-ge\-Ad\-di\-tion Problem for $G$
and $uv$ which results from making $u$ or $v$ universal in $G$.
\end{lemma}
\begin{proof}
Let $B=V(G)\setminus A$. By considering the fill edges with one
endpoint in $A$ and the other in $B$, the number~$N$ of fill edges
(including $uv$) satisfies
\[
\begin{split}
N
&\ge
|A\cap C_u|\cdot|B\cap C_v|
+
|A\cap C_v|\cdot|B\cap C_u|\\
&\ge
\bigl(|A\cap C_u|+|B\cap C_v|-1\bigr)
+
\bigl(|A\cap C_v|+|B\cap C_u|-1\bigr)\\
&=
|V(G)|-2,
\end{split}
\]
by Equation~\ref{eq:product}.

On the other hand, making $u$ universal in $G$ gives a  solution requiring $|V(G)|-1-\deg_G(u)$ fill edges, including $uv$ (Eq.~\ref{eq:u_universal}). Since $H$ is optimal, $N\le |V(G)|-1-\deg_G(u)$. Combining this with $N\ge |V(G)|-2$ yields $\deg_G(u)\le 1$. Since $G[C_u]$ is connected and $|C_u|\ge2$, we have $\deg_G(u)\ge1$, so $\deg_G(u)=1$; symmetrically $\deg_G(v)=1$. Making $u$ (or $v$) universal then requires exactly $|V(G)|-2$ fill edges, matching the lower bound, so such a solution is optimal.
% Since $A\cap C_u\ne\emptyset$ and $B\cap C_u\ne\emptyset$, we have $|C_u|\ge2$. Moreover, $G[C_u]$ is connected by Observation~\ref{obs: solution_prop}(i), and therefore $\deg_G(u)\ge1$. Hence $\deg_G(u)=1$. By symmetry, $\deg_G(v)=1$. Therefore, making either $u$ or $v$ universal in $G$ requires exactly $|V(G)|-2$
% fill edges. Since $N\ge |V(G)|-2$ and $H$ is optimal, such a solution
% is also optimal. Thus, there exists an optimal solution of the
% ($P_4$-spar\-se-2CC,$+1$)-Min\-Ed\-ge\-Ad\-di\-tion Problem for $G$
% and $uv$ in which $u$ or $v$ is universal in $G$.
\end{proof}

Our algorithm takes into account the structure of the $P_4$-sparse tree of the given graph. In particular, for a $P_4$-sparse graph~$G$ with at least two vertices and a vertex~$u$ of $G$, let $t_1 t_2 \cdots t_{j_u}$ be the path in the $P_4$-sparse tree~$T_G$ of $G$ from the root~$t_1$ to the leaf~$t_{j_u}$ corresponding to $u$. Then, we define the subtrees $T_{u,1}(G)$, $T_{u,2}(G)$, $\ldots$, $T_{u,j_u-1}(G)$ as follows:
\begin{itemize}
	\item
	$T_{u,1}(G)$ is the subtree of $T_G$ containing $t_1$ after we have removed the tree edge~$t_1 t_2$;
	\item
	for $j = 2, 3, \ldots, j_u-1$, $T_{u,j}(G)$ is the subtree of $T_G$ containing $t_j$ after we have removed the tree edges~$t_{j-1} t_j$ and $t_j t_{j+1}$.
\end{itemize}
Figure~\ref{fig:trees}(b) depicts the path $t_1 t_2 \cdots t_j t_{j+1} \cdots u$ and the subtrees $T_{u,1}(G)$, $T_{u,2}(G)$, $\ldots$, $T_{u,j}(G)$, $T_{u,j+1}(G),\ldots$. Moreover, the fact that $C_u,C_v$ are connected components of $G$ implies that:

\begin{observation} \label{obs: solution_prop}
Let $G$ be a disconnected graph consisting of $2$~connected components $C_u$ and $C_v$ such that $u \in C_u$ and $v \in C_v$, and consider 
%the instance of 
the ($P_4$-sparse-$2$CC,$+1$)-Min\-Edge\-Addition Problem for the graph~$G$ and the non-edge $uv$. Then
\begin{itemize}
	\item[(i)] each of the induced subgraphs~$G[C_u]$ and $G[C_v]$ is connected;
	\item[(ii)] In any optimal solution~$H$ of the ($P_4$-sparse-$2$CC,$+1$)-Min\-Ed\-ge\-Ad\-di\-tion Problem for the graph~$G$ and the added non-edge $uv$, it holds that each of the induced subgraphs $H[C_u]$ and $H[C_v]$ is connected.
    %; thus, the entire graph~$H$ is connected and the root node of the $P_4$-sparse tree of $H$ is a $1$-node or a $2$-node.
\end{itemize}
\end{observation}

%Moreover, we can show the following result.

%\begin{lemma} \label{lemma: solution_connected}
%Let $G$ be a disconnected graph consisting of $2$~connected components $C_u$ and $C_v$ such that $u \in C_u$ and $v \in C_v$, and consider the instance of the ($P_4$-sparse-$2$CC,+$1$)-Min\-Edge\-Addition Problem for the graph~$G$ and the non-edge $uv$. In any optimal solution~$H$ of the ($P_4$-sparse-2CC,{$+$}$1$)-Min\-Edge\-Addition Problem for $G$ and $uv$ and for any subset~$S \subseteq V(G)$, it holds that if the induced subgraph~$H[S]$ is connected then each of the induced subgraphs $G[S \cap C_u]$ and $G[S \cap C_v]$ is also connected.
%\end{lemma}
%
%\begin{proof}
%Suppose for contradiction that the induced subgraph $G[S \cap C_u]$ is disconnected and let $C'$ 
%\end{proof}

\noindent
Observation~\ref{obs: solution_prop}(ii) implies that any solution~$H$ of the ($P_4$-sparse-2CC,$+1$)-Min\-Edge\-Addition Problem for $G$ and $uv$ is connected and thus the root node of its $P_4$-sparse tree~$T_H$ is a $1$-node or a $2$-node (for a thin or a thick spider); these are the cases that we consider below.

%\subsection{The root
%\subsection{Case 1: The root node 
\medskip\noindent
\textbf{Case 1: The root
of the $P_4$-sparse tree~$T_H$ of a solution~$H$ is a $1$-node.}

\noindent
%If the leaf corresponding to $u$ ($v$ resp.) in the $P_4$-sparse tree~$T_H$ of a solution~$H$ of the ($P_4$-sparse-2CC,$+1$)-Min\-Edge\-Addition Problem for $G$ and $uv$ is a child of the root of $T_H$, then $u$ ($v$ resp.) is universal in $H$. So, in the following, assume that the leaves corresponding to $u, v$ are not children of $T_H$'s root. 
Let $T_u, T_v$ be the subtrees rooted at children of the root of 
%$T_H$ 
the $P_4$-sparse tree~$T_H$ of a solution~$H$ of the ($P_4$-sparse-2CC,$+1$)-Min\-Edge\-Addition Problem for $G$ and $uv$ 
containing the leaves corresponding to $u$ and $v$, respectively. Next, we consider the cases whether $T_u = T_v$ and $T_u \ne T_v$.

\smallskip\noindent
\emph{Case 1a. The leaves for $u,v$ belong to subtrees~$T_u, T_v$, resp., of $T_H$ and $T_u = T_v$.}
%
%We show the following lemma.

\begin{lemma} \label{lemma:1a_1}
Suppose that the root node of the $P_4$-sparse tree~$T_H$ of an optimal solution~$H$ of the ($P_4$-sparse-2CC,{$+$}$1$)-Min\-Edge\-Addition Problem for the graph~$G$ and the non-edge~$uv$ is a $1$-node and that for the subtrees $T_u, T_v$ of $T_H$ it holds that $T_u = T_v$.
Then, there exists an optimal solution~$H'$ of the ($P_4$-spar\-se-2CC,{$+$}$1$)-Min\-Edge\-Addition Problem for the graph~$G$ and the non-edge~$uv$
\begin{itemize}
\item[(a)]
which results from making $u$ or $v$ universal in $G$ \quad or
\item[(b)]
in which the subtree~$T_{u,1}(H')$ is identical to $T_{u,1}(G[C_u])$ or $T_{v,1}(G[C_v])$ \quad or
\item[(c)]
$G[C_u]$ is a thin spider $(S,K,R)$ with $u \in S$, in which case $H'$ results from making the unique neighbor of $u$ universal in $G$ for a total of $|K|-1+|C_v|+\mu_u$ fill edges where $\mu_u$ is the number of fill edges (including $uv$) in an optimal solution of the ($P_4$-sparse-2CC,{$+$}$1$)-Min\-Ed\-ge\-Ad\-di\-tion Pro\-blem for the 
%disconnected 
graph~$G[\{u\} \cup C_v]$ and the non-edge~$uv$ or a similar case holds for $G[C_v]$.
\end{itemize}
\end{lemma}
\begin{proof}
Since $T_u = T_v$, it holds that neither $u$ nor $v$ is universal in $H$, that $|C_u| + |C_v| \ge 3$, and that $T_{u,1}(H) = T_{v,1}(H)$; moreover, since $u,v \in V(T_u)$, Lemma~\ref{lemma:useful} implies that if $V(T_{u,1}(H)) \cap C_u \ne \emptyset$ and $V(T_{u,1}(H)) \cap C_v \ne \emptyset$, there exists an optimal solution of the ($P_4$-spar\-se-2CC,{$+$}$1$)-Min\-Edge\-Addition Problem for $G$ and $uv$ by making $u$ or $v$ universal in $G$. So, in the following, suppose without loss of generality that $V(T_{u,1}(H)) \subset C_u$.

We distinguish the following cases depending on the treenode type of the root of the subtree~$T_{G[C_u]}$ of the induced subgraph~$G[C_u]$ of $G$; since $G[C_u]$ is connected (Observation~\ref{obs: solution_prop}(i)), the root of $T_{G[C_u]}$ is a $1$-node or a $2$-node.
\begin{enumerate}
\item [A.] \emph{The root of the subtree~$T_{G[C_u]}$ is a $1$-node}. If $V(T_{u,1}(G[C_u])) \subseteq V(T_{u,1}(H))$, then in $H$, every vertex in $V(T_{u,1}(G[C_u]))$ is adjacent to all the vertices in $V(G) \setminus V(T_{u,1}(G[C_u]))$; then, the join $H' = G[V(T_{u,1}(G[C_u]))] + H[V(G) \setminus V(T_{u,1}(G[C_u]))]$ has no more fill edges than $H$ and thus is an optimal solution with $T_{u,1}(H') = T_{u,1}(G[C_u])$.
		
Let $Q = V(T_{u,1}(G[C_u])) \setminus V(T_{u,1}(H))$ and consider now the case in which $Q \ne \emptyset$.
%; in particular, assume that $H$ is such that $|Q|$ is minimum.
Then, $Q$ is universal in $G[C_u \setminus V(T_{u,1}(G[C_u]))]$, which includes $u$. Additionally, since $u$ is adjacent to all the vertices in $V(T_{u,1}(G[C_u]))$, the graph~$G[C_u \setminus V(T_{u,1}(H))]$ is connected and so is the induced subgraph $H[V(G) \setminus V(T_{u,1}(H))]$; since the root of the $P_4$-sparse tree of $H$ is a $1$-node, then any non-leaf child of the root is a $2$-node.
		
\smallskip\noindent
\emph{If the least common ancestor~$t$ of the leaves for $u,v$ is not a child of the root of the $P_4$-sparse tree~$T_H$ of $H$}, then the subtree~$T_{u,2}(H)$ is well defined and its root is a $2$-node; let $(S_1,K_1,R_1)$ be the corresponding spider and $u,v \in R_1$.
	\begin{itemize}
	\item
	\emph{The spider~$(S_1,K_1,R_1)$ is thin}. If $S_1 \cup K_1$ contains vertices from both $C_u$ and $C_v$, Lemma~\ref{lemma:2a_2} implies that there exists an optimal solution~$H'$ of the ($P_4$-sparse-2CC,{$+$}$1$)-Min\-Edge\-Addition Problem for the subgraph $G[V(G) \setminus V(T_{u,1}(H))]$ and the non-edge~$u v$ in which $u$ or $v$ is universal in $H'[V(G) \setminus V(T_{u,1}(H))]$ or the induced subgraphs $G[C_u \setminus V(T_{u,1}(H))]$ and $G[C_v]$ are as shown in Figure~\ref{fig:trees}(c).
    In the former case, since $u$ or $v$, resp., are adjacent to the vertices in $V(T_{u,1}(H))$ as well, they are universal in $H'$.  
    In the latter case, vertex~$u'$ in Figure~\ref{fig:trees}(c) is the only universal vertex in $G[C_u \setminus V(T_{u,1}(H))]$ and thus $Q = \{u'\}$. If $|R_1 \cap C_v| \le |R_1 \cap C_u|$, there exists an optimal solution~$H''$ of the ($P_4$-sparse-2CC,$+1$)-Min\-Ed\-ge\-Ad\-di\-tion Problem for $G[V(G) \setminus V(T_{u,1}(H))]$ and $uv$ in which $u'$ is universal in $H''[V(G) \setminus V(T_{u,1}(H))]$, in which case we move $u'$ to $T_{u,1}(H)$ and have $T_{u,1}(H'') = T_{u,1}(G[C_u])$.
    Now, if $|R_1 \cap C_u| < |R_1 \cap C_v|$, there exists an optimal solution~$H''$ of the ($P_4$-sparse-2CC,$+1$)-Min\-Ed\-ge\-Ad\-di\-tion Problem for $G[V(G) \setminus V(T_{u,1}(H))]$ and $uv$ in which $v'$ is universal in $H''[(C_u \setminus V(T_{u,1}(H))) \cup C_v]$; but then $v'$ is universal in $H''$ and we have that $T_{v,1}(H'') = T_{v,1}(G[C_v])$.
    			
	Let us now consider that $S_1 \cup K_1 \subseteq C_u$ or $S_1 \cup K_1 \subseteq C_v$. However, it is not possible that $S_1 \cup K_1 \subseteq C_u$, otherwise $S_1 \cap Q = \emptyset$ (no vertex in $S_1$ is adjacent to $u$) and no vertex is adjacent to all the vertices in $G[S_1]$, which implies that $Q = \emptyset$, a contradiction. Hence $S_1 \cup K_1 \subseteq C_v$ but then by exchanging $T_{u,1}(H)$ and $T_{u,2}(H)$, we get a $P_4$-sparse completion of $G+uv$ with fewer fill edges than in $H$, a contradiction.
	\item
	\emph{The spider~$(S_1,K_1,R_1)$ is thick}. Lemma~\ref{lemma:3a_2} implies that either $S_1 \cup K_1 \subseteq C_u$ or $S_1 \cup K_1 \subseteq C_v$; as previously, the former case is impossible otherwise $S_1 \cap Q = \emptyset$ (no vertex in $S_1$ is adjacent to $u$) and no vertex is adjacent to all the vertices in $G[S_1]$ which implies that $Q = \emptyset$, a contradiction. In the latter case, by exchanging $T_{u,1}(H)$ and $T_{u,2}(H)$, we get an optimal solution with fewer fill edges than in $H$, a contradiction.
	\end{itemize}
\emph{If the least common ancestor~$t$ of $u,v$ is a child of the root of the $P_4$-sparse tree~$T_H$ of $H$}, then $t$ is a $2$-node. If $G[C_u \setminus V(T_{u,1}(H))]$ is a $P_2$ then $|C_u| \ge 3$ and $u$ is universal in $G[C_u]$. Then the number of fill edges in $H$ (including $uv$) is at least $(|C_u| - 2) \cdot |C_v| + 1 \ge |C_v| + 1$, a contradiction to the optimality of $H$ since $u$ is universal in $G[C_u]$ and by using $|C_v|$ fill edges it becomes universal in $G$.
On the other hand, $G[C_u \setminus V(T_{u,1}(H))]$ is not a spider since no subset of vertices are adjacent to all remaining vertices in a spider. Let $(S_2,K_2,R_2)$ be the spider in $H$ corresponding to the treenode~$t$.
	\begin{itemize}
	\item
	\emph{The spider~$(S_2,K_2,R_2)$ is thin}.
	If one of $u,v$ belongs to $S_2 \cup K_2$ and the other belongs to $R_2$, the subgraph~$G[C_u]$ cannot be a $P_3$ with $u$ as an endpoint and Lemma~\ref{lemma:2b} implies that there exists an optimal solution~$H'$ with $u$ or $v$ being universal in $H'[V(G) \setminus V(T_{u,1}(H))]$, $T_{u,1}(H') = T_{u,1}(G[V(G) \setminus V(T_{u,1}(H))])$, or $T_{v,1}(H') = T_{v,1}(G[C_v])$. If $T_{u,1}(H') = T_{u,1}(G[V(G) \setminus V(T_{u,1}(H))])$ then because $Q \ne \emptyset$, the root node of $T_{u,1}(H')$ is a $1$-node and then $Q \subseteq V(T_{u,1}(H'))$ and thus we can move them to $T_{u,1}(H)$ to have an optimal solution~$H''$ with $T_{u,1}(H'') = T_{u,1}(G[C_u])$. If $T_{v,1}(H') = T_{v,1}(G[C_v])$ then we can exchange $T_{u,1}(H)$ and $T_{v,1}(H')$ getting an optimal solution~$H''$ with $T_{v,1}(H'') = T_{v,1}(G[C_v])$.
	On the other hand, if $u,v$ in $S_2 \cup K_2$, since $G[C_u \setminus V(T_{u,1}(H))]$ is neither a $P_2$ nor a thin spider, then Lemma~\ref{lemma:2c} implies that there exists an optimal solution~$H'$ with $u$ or $v$ being universal in $H'[V(G) \setminus V(T_{u,1}(H))]$.
	\item
	\emph{The spider~$(S_2,K_2,R_2)$ is thick}. 
	Since at least one of $u,v$ belongs to $S_2 \cup K_2$, Lemma~\ref{lemma:3bc} implies that there exists an optimal solution~$H'$ with $u$ or $v$ being universal in $H'[V(G) \setminus V(T_{u,1}(H))]$.
\end{itemize}
Therefore, if the least common ancestor~$t$ of $u,v$ is a child of the root of $T_H$ there exists an optimal solution~$H'$ with $u$ or $v$ being universal in $H'[V(G) \setminus V(T_{u,1}(H))]$ which implies that there exists an optimal solution~$H''$ ($P_4$-sparse-2CC,{$+$}$1$)-Min\-Edge\-Addition Problem for the graph~$G$ and the non-edge~$uv$ with $u$ or $v$ being universal in $H''$.
		
\item [B.] \emph{The root of the subtree~$T_{G[C_u]}$ is a 2-node corresponding to a thin or a thick spider $(S_G,K_G,R_G)$}. Let $S_G = \{ s_1, \ldots, s_{|K_G|} \}$ and $K_G = \{ k_1, \ldots, k_{|K_G|} \}$ where $s_i, k_i$ ($i = 1, \ldots, |K_G|$) are adjacent (non-adjacent resp.) if $G[C_u]$ is a thin (thick, resp.) spider.
If $K_G \not\subseteq V(T_{u,1}(H))$ and there exist vertices in $V(T_{u,1}(H)) \setminus (S_G \cup K_G)$, then we exchange vertices in $K_G \setminus V(T_{u,1}(H))$ with vertices in $V(T_{u,1}(H)) \setminus (S_G \cup K_G)$; note that for any vertex~$k_i \in K_G$ and any vertex~$w \in C_u \setminus (S_G \cup K_G)$, it holds that $N_G[w] \subseteq N_G[k_i]$. Additionally, for any $i$ ($i = 1, \ldots, |K_G|$) such that $s_i \in V(T_{u,1}(H))$ and $k_i \not\in V(T_{u,1}(H))$, we exchange $s_i$ and $k_i$; note that again $N_G[s_i] \subseteq N_G[k_i]$. After these exchanges, which do not increase the number of fill edges, we have constructed an optimal solution~$H'$ of the ($P_4$-sparse-2CC,{$+$}$1$)-MinEdgeAddition Problem for the graph~$G$ and the non-edge~$uv$, and in $H'$, there is no vertex~$s_j \in V(T_{u,1}(H'))$ such that $k_j \not\in V(T_{u,1}(H'))$.
		
Let $K' \subseteq K_G$ be the set of vertices $k_j \in V(T_{u,1}(H))$ such that $s_j \ne u$ (it is possible that $u \in S_G$) and let $S' = \{ \,s_j \,| \,k_j \in K' \,\}$. Then, if $|K'| = 1$ and $K' = \{k_j\}$, let $F$ be the graph resulting from $H'$ after having removed the vertices $s_j, k_j$, having inserted $s_j$ as an isolated vertex, and after having made $k_j$ as a universal vertex whereas if $|K'| \ge 2$, let $F$ be the spider $(S',K',V(G) \setminus (S' \cup K')$ where $F[V(G) \setminus (S' \cup K')] = H'[V(G) \setminus (S' \cup K')]$ (the spider is thin or thick if and only if the spider $(S_G,K_G,R_G)$ is thin or thick respectively). In either case, the graph~$F$ is $P_4$-sparse and is a completion of $G$ including the non-edge~$uv$ and has fewer fill edges than $H$ if $S' \ne \emptyset$. Then, the optimality of $H$ implies that the only possibility is that $V(T_{u,1}(H)) = \{k_j\}$ such that $s_j = u$.
\end{enumerate}
%\hfill $\square$
\end{proof}

%\begin{figure}[h!]
%\centerline{\includegraphics[scale=0.9]{proof.eps}}
%\caption{(a)... ,(b)}
%\label{fig:proof}
%\end{figure}

\noindent
\emph{Case 1b. The leaves for $u,v$ belong to subtrees~$T_u,T_v$, resp., of $T_H$ with $T_u \ne T_v$.}
%Then, we show the following lemma.

\begin{lemma} \label{lemma:1b_1}
	Let $H$ be an optimal solution of the ($P_4$-sparse-2CC,{$+$}$1$)-Min\-Edge\-Addition Problem for the graph~$G$ and the non-edge~$uv$ and suppose that the vertices $u,v$ belong to subtrees~$T_1,T_2$, respectively, of the root of the $P_4$-sparse tree~$T_H$ of $H$.
	%If $A = V(T_1)$ and $B = V(G) \setminus V(T_1)$, then it is not possible that $A \cap C_v \ne \emptyset$ and $B \cap C_u \ne \emptyset$ and 
 Then there exists an optimal solution of the ($P_4$-sparse-2CC,{$+$}$1$)-Min\-Edge\-Addition Problem for the graph~$G$ and the non-edge~$uv$ which results from making $u$ or $v$ universal in $G$.
\end{lemma}

\begin{proof}
Set $A=V(T_1)$ and $B=V(G)\setminus A$. By Lemma~3, if $A\cap C_v\ne\emptyset$ and $B\cap C_u\ne\emptyset$, we are done. Otherwise, by symmetry, assume $A\cap C_v=\emptyset$. If $A=\{u\}$, then $u$ is universal in $H$ and we are done. Suppose next that $|A| \ge 2$. Then by considering the fill edges with one endpoint in $A$ and the other in $B$, we get that the number~$N$ of fill edges in $H$ (including $uv$) is
\begin{align*}
N \ \ge\  & |B \setminus N_G(u)| +
|A \setminus \{u\}| \cdot |C_v|
\ \ge\  |B \setminus N_G(u)| + |A| - 1 + |C_v| - 1\\
\ge\  & |V(G) \setminus N_G[u]| + |C_v| - 1
\ \ge\  |V(G) \setminus N_G[u]|
\end{align*}
which implies that there is an optimal solution which results from making $u$ universal in $G$.
\end{proof}

%\subsection{
\medskip\noindent
\textbf{Case 2: The root node of the $P_4$-sparse tree~$T_H$ of the solution~$H$ is a $2$-node corresponding to a thin spider $(S,K,R)$.}
\ We first prove some important properties for the optimal solution~$H$ in this case.

\begin{observation} \label{obs:2a}
	Suppose that an optimal solution $H$ of the $(P_4\text{-sparse-2CC},+1)$-MinEdgeAddition Problem for a $P_4$-sparse graph $G$ and a non-edge $uv$ is a thin spider $(S,K,R)$, and let $f \colon S \to K$ be its matching bijection. Then:

\begin{enumerate}[(i)]
    \item For each edge $ab \in E(H)$ such that $a \in K$, $b \in S$, and $b \notin \{u,v\}$, it holds that $ab \in E(G)$.
    \item Let $x \in \{u,v\}$, let $C_x$ be the connected component of $G$ containing $x$, and let $a,c \in K \cap C_x$. If
    $
    f^{-1}(a), f^{-1}(c) \in C_x,
    $
    then
    $
    ac \in E(G).
    $
    \item Consequently, if
    $
    a,c \in K \cap C_x
    $
    and
    $
    ac \in E(H) \setminus E(G),
    $
    then at least one of $f^{-1}(a)$ and $f^{-1}(c)$ is the endpoint of the prescribed non-edge that does not belong to $C_x$.
\end{enumerate}
\end{observation}
\begin{proof}

    For (i), suppose that $ab \notin E(G)$. Since $b \in S$, the vertex $b$ has exactly one neighbor in $H$, namely $a$. Delete the fill edge $ab$ from $H$. The resulting graph is the disjoint union of the isolated vertex $b$ and the graph $H - b$, and is therefore $P_4$-sparse. Moreover, since $b \notin \{u, v\}$, deleting $ab$ does not delete the prescribed edge $uv$. Thus we obtain a $P_4$-sparse completion with fewer fill edges than $H$, contradicting the optimality of $H$. Hence $ab \in E(G)$.
    
    For (ii), let
    $
    a' = f^{-1}(a)$ and $c' = f^{-1}(c)$. By assumption,
    $
    a', c' \in C_x.
    $
    Since $a'$ and $c'$ belong to the stable set of the thin spider $H$, their only neighbors in $H$ are $a$ and $c$, respectively. Because $G[C_x]$ is connected and $G$ is a spanning subgraph of $H$, it follows that
    $
    a'a, \; c'c \in E(G)$. If $ac \notin E(G)$, any shortest path from $a'$ to $c'$ in $G[C_x]$ starts with $a'a$ and ends with $cc'$ and has length at least $4$, so its first five vertices form a chordless $P_5$, which contains two distinct induced $P_4$'s; contradicting $P_4$-sparseness. Hence $ac \in E(G)$.
    
    % Suppose, for contradiction, that
    % $
    % ac \notin E(G).
    % $
    % Consider a shortest path $P$ from $a'$ to $c'$ in the connected graph $G[C_x]$. The path starts with the edge $a'a$ and ends with the edge $cc'$. Since $ac \notin E(G)$, the path has at least four edges.
    
    % Every shortest path is chordless. Hence the first five consecutive vertices of $P$ induce a $P_5$. But a $P_5$ contains two distinct induced $P_4$'s on the same set of five vertices, contradicting the fact that $G$ is $P_4$-sparse. Therefore
    % $
    % ac \in E(G).
    % $
    
    For (iii), suppose that
    $
    ac \in E(H) \setminus E(G)
    $
    with $a, c \in K \cap C_x$. If neither $f^{-1}(a)$ nor $f^{-1}(c)$ were the endpoint outside $C_x$, then both matched stable-set vertices would belong to $C_x$. Indeed, by part (i), every leg whose stable-set endpoint is different from $u$ and $v$ is an edge of $G$, and therefore its two endpoints belong to the same connected component of $G$. Part (ii) would then imply that $ac \in E(G)$, a contradiction. This proves (iii).
    
	% (i)~Suppose without loss of generality that $a, b$ are not adjacent in $G$; then, $ab$ is a fill edge in $H$. Let $H'$ be the graph resulting from $H$ after we have removed the edge~$ab$. The graph~$H'$ is $P_4$-sparse since it is the union of the isolated vertex~$b$ with the induced subgraph~$H[V(G) \setminus \{b\}]$;
 %    %in fact, 
 %    since $b$ is not $u$ or $v$, $H'$ is a $P_4$-sparse completion of the graph~$G+uv$ that has $1$ fewer fill edge than $H$, in contradiction to the optimality of $H$. Therefore, $a, b$ are adjacent in $G$.
	
	% \medskip\noindent
	% (ii)~We concentrate only in the case in which $a, c \in K \cap C_u$. In $H$, let $a'$ ($c'$, resp.) be the unique neighbor of $a$ ($c$, resp.) in $S$; by statement~(i) of this observation, $a',c' \in C_u$, and $a, a'$ and $c, c'$ are adjacent in $G$.
	% Now, suppose, for contradiction, that $a,c$ are not adjacent in $G$. Since $a,c \in C_u$ and the induced subgraph~$G[C_u]$ is connected (Observation~\ref{obs: solution_prop}(i)), there is a path connecting $a'$ to $c'$ in $G[C_u]$, and in fact there is a chordless such path~$\rho$. Clearly, $\rho$ starts with the edge~$a' a$, ends at the edge~$c c'$ and has length at least~$4$; thus, $G$ contains an induced chordless path on at least~$5$ vertices, in contradiction to the fact that $G$ is $P_4$-sparse.
\end{proof}

%\newpage

\noindent
\emph{Case 2a. The vertices $u,v$ belong to $R$.} 
We first consider the case in which $S \cup K \subset C_u$ or $S \cup K \subset C_v$. For this case, we show the following lemma.

\begin{lemma} \label{lemma:2a_1}
Suppose that the optimal solution~$H$ of the
($P_4$-sparse-2CC,{$+$}$1$)-Min\-Edge\-Addition Problem for a
$P_4$-sparse graph~$G$ and a non-edge~$uv$ is a thin spider
$(S,K,R)$ with $u,v\in R$.
If $S\cup K\subseteq C_u$, then there exists an optimal solution~$H'$
of the ($P_4$-sparse-2CC,{$+$}$1$)-Min\-Edge\-Addition Problem
for the graph~$G$ and the non-edge~$uv$
\begin{itemize}
\item[(i)] which results from making $u$ or $v$ universal in $G$, or
\item[(ii)] in which $T_{u,1}(H')$ is identical to
$T_{u,1}(G[C_u])$ or $T_{v,1}(H')$ is identical to
$T_{v,1}(G[C_v])$.
\end{itemize}
A symmetric result holds if $S\cup K\subseteq C_v$.
\end{lemma}

\begin{proof}
	We consider the following cases that cover all possibilities:
	\begin{enumerate}
		\item[A.] \emph{The root node of the subtree~$T_{G[C_u]}$ is a $1$-node}. This implies that every vertex in $V(T_{u,1}(G[C_u]))$ is adjacent to all vertices in $C_u \setminus V(T_{u,1}(G[C_u]))$ and in particular to $u$. On the other hand, the vertices in $S$ are not adjacent to $u$ in $H$ and consequently are not adjacent to $u$ in $G$; hence, since $S \subset C_u$, it holds that $S \subset C_u \setminus V(T_{u,1}(G[C_u]))$ which in turn implies that in $G$, all the vertices in $V(T_{u,1}(G[C_u]))$ are adjacent to all the vertices in $S$ and this is also true in $H$. But this is \emph{impossible} since no vertex in $H$ is adjacent to all vertices in $S$.
		
		\item[B.] \emph{The root node of the subtree~$T_{G[C_u]}$ is a $2$-node corresponding to a thin spider $(S_G,K_G,R_G)$}.
		Since each vertex in $C_u \setminus S_G$ has degree at least~$2$ in $G$ and thus it has degree at least $2$ in $H$, and each vertex in $S \subset C_u$ has degree~$1$, we conclude that $S \subseteq S_G$. Then, by Observation~\ref{obs:2a}(i), $K = N_G(S)$ and $K \subseteq K_G$. If $K = K_G$ then $S = S_G$ and $T_{u,1}(H) = T_{u,1}(G[C_u])$.
		
		In the following, assume that $K \subset K_G$. Then if $|K_G \setminus K| \ge 2$, the subgraph~$G[C_u \setminus (S \cup K)]$ is a thin spider $(S_G \setminus S, K_G \setminus K, R_G)$ whereas if $K_G \setminus K = \{w\}$ then $w$ is universal in $G[C_u \setminus (S \cup K)]$ and the remaining vertices form a disconnected graph with connected components $R_G$ and $z$ where $\{z\} = N_G(w) \cap S_G$. In either case, $|C_u \setminus (S \cup K)| \ge 3$ and $G[C_u \setminus (S \cup K)]$ is connected.
		
		If the least common ancestor~$t$ of $u,v$ is not a child of the root of the $P_4$-sparse tree~$T_H$ of $H$, then the subtree~$T_{u,2}(H)$ is well defined and its root is a $1$-node or a $2$-node.
		\begin{itemize}
			\item[(a)]
			\emph{The root node of $T_{u,2}(H)$ is a $1$-node}.
			Then Lemma~\ref{lemma:1a_1} implies that there is an optimal solution~$F$ of the ($P_4$-sparse-2CC,{$+$}$1$)-Min\-Edge\-Addition Problem for the subgraph $G[V(G) \setminus (S \cup K)]$ and the non-edge~$u v$ in which either $T_{u,1}(F) = T_{u,1}(G[C_u \setminus (S \cup K)])$ or $T_{v,1}(F) = T_{v,1}(G[C_v])$ or $u$ or $v$ is universal. The former case is impossible since by replacing $H[S \cup K \cup V(T_{u,1}(F))]$ by $G[S \cup K \cup V(T_{u,1}(F))]$, we get an optimal solution with fewer fill edges in contradiction to the optimality of $H$.  %%CHECK
			\item[(b)]
			\emph{The root node of $T_{u,2}(H)$ is a $2$-node}.
			Let $(S_1,K_1,R_1)$ be the corresponding spider and $u,v \in R_1$.\\
			\textbf{The spider~$\boldsymbol{(S_1,K_1,R_1)}$ is thin}. If $S_1 \cup K_1$ contains vertices from both $C_u$ and $C_v$, Lemma~\ref{lemma:2a_2} implies that there exists an optimal solution~$H'$ of the ($P_4$-sparse-2CC,{$+$}$1$)-Min\-Edge\-Addition Problem for the subgraph $G[V(G) \setminus V(T_{u,1}(H))]$ and the non-edge~$u v$ in which either $u$ or $v$ is universal in $H'$ or $T_{u,1}(H')$ is identical to $T_{u,1}(G[V(G) \setminus V(T_{u,1}(H))])$ or $T_{v,1}(G[C_v])$. In the latter case, by exchanging $T_{u,1}(H')$ and $T_{u,2}(H')$ we get an optimal solution of the ($P_4$-sparse-2CC,{$+$}$1$)-Min\-Edge\-Addition Problem for $G$ and $u v$ in which $T_{u,1}(H')$ is identical to $T_{u,1}(G[C_u])$ or $T_{v,1}(G[C_v])$. In turn, if vertex~$u$ or $v$ is universal in an optimal solution of the ($P_4$-sparse-2CC,{$+$}$1$)-Min\-Edge\-Addition Problem for the induced subgraph $G[V(G) $ $ \setminus (S \cup K)]$ and $uv$, then there exists an optimal solution of the ($P_4$-sparse-2CC,{$+$}$1$)-Min\-Edge\-Addition Problem for $G$ and $uv$ in which $u$ or $v$ is universal; note that solution~$H$ contains fill edges connecting the vertices in $K$ to all the vertices in $(S_G \setminus S) \cup C_v$, which, for $|C_v| \ge 2$, are more than the $|K|$ fill edges needed to connect $u$ or $v$ to the vertices in $S$.
				
				If $S_1 \cup K_1 \subseteq C_u$ then $S_1 \subseteq S_G$ which implies that $K_1 \subseteq K_G$, and if we replace $H[(S \cup S_1) \cup (K \cup K_1)]$ by $G[(S \cup S_1) \cup (K \cup K_1)]$ we get a $P_4$-sparse completion of the graph~$G+uv$ with fewer fill edges than $H$, a contradiction. Hence $S_1 \cup K_1 \subseteq C_v$. Then, because $|K| \ge 2$, $|K_1| \ge 2$, we have that $|C_u| \ge 5$ and $|C_v| \ge 5$, and the number~$N$ of fill edges is (Equation~\ref{eq:product})
				\begin{align*}
				N \ \ge \  & |K| \cdot |C_v| + |K_1| \cdot |C_u \setminus (S \cup K)|\\
				\ \ge\  & |C_v| + (|K| - 1) \cdot |C_v| + 2\, |C_u \setminus (S \cup K)|\\
				\ =\ & |C_v| + (|K| - 1) \cdot (|C_v| - 4) + 4\, (|K| - 1) + 2\, |C_u| - 4\, |K|\\
				\ \ge\ & 2\, |C_v| - 4 + 2\, |C_u| - 4
				\ \ge\ |C_v| + |C_u| + 2
				\end{align*}
				which is greater than making $u$ or $v$ universal, a contradiction to the optimality of $H$.\\
				\textbf{The spider~$\boldsymbol{(S_1,K_1,R_1)}$ is thick}. Lemma~\ref{lemma:3a_2} implies that either $S_1 \cup K_1 \subseteq C_u$ or $S_1 \cup K_1 \subseteq C_v$. If $S_1 \cup K_1 \subseteq C_v$ then by working as in the previous case, we get a contradiction. If $S_1 \cup K_1 \subseteq C_u$ then because the vertices in $S_1$ are not adjacent to $u$, we get that $S_1 \subseteq S_G$; let $z \in S_1$ and let $w$ be the unique neighbor of $z$ in $K_G$. Then because $w \not\in K$ and because the only neighbors of $z$ in $H[V(G) \setminus (S \cup K)$ belong to $K_1$, we have that $w \in K_1$. But then, by removing $w,z$ from the spider~$(S_1,K_1,R_1)$ and joining them to the spider~$(S,K,R)$, we get a $P_4$-sparse completion of $G+uv$ that requires fewer fill edges than $H$, a contradiction.
			%\end{itemize}
		\end{itemize}
		
		If the least common ancestor~$t$ of $u,v$ is a child of the root of the $P_4$-sparse tree~$T_H$ of $H$, then $t$ is a $1$-node or a $2$-node. 
		\begin{itemize}
			\item[(a)]
			\emph{The root node of $T_{u,2}(H)$ is a $1$-node}.
			Then Lemma~\ref{lemma:1b_1} implies that there exists an optimal solution~$F$ of the ($P_4$-sparse-2CC,{$+$}$1$)-Min\-Edge\-Addition Problem for the subgraph $G[V(G) \setminus (S \cup K)]$ and the non-edge~$u v$ which results from making $u$ or $v$ universal in $G[V(G) \setminus (S \cup K)]$.
			\item[(b)]
			\emph{The root node of $T_{u,2}(H)$ is a $2$-node}.
			Let $(S_2,K_2,R_2)$ be the spider corresponding to the treenode~$t$.
			\begin{itemize}
				\item
				\emph{The spider~$(S_2,K_2,R_2)$ is thin}.
				If one of $u,v$ belongs to $S_2 \cup K_2$ and the other belongs to $R_2$, Lemma~\ref{lemma:2b} applies. If Lemma~\ref{lemma:2b}, case~(c) holds, then because $|C_u \setminus (S \cup K)| \ge 3$, we conclude that $G[C_v]$ is a $P_2$ and $G[C_u \setminus (S \cup K)]$ is a $P_3$ with $u$ as an endpoint or a thin spider with $u$ being an isolated vertex in the spider's $R$-set; let the resulting spider be $(S',K',R')$. Then, we can get a $P_4$-sparse completion of the graph~$G+uv$, which is a spider with stable set $S \cup (S' \cap C_u)$ and clique $K \cup (K' \cap C_u)$, requiring fewer fill edges than in $H$, a contradiction. A similar construction implies that Lemma~\ref{lemma:2b}, case~(b) if $T_{u,2}(H) = T_{u,1}(G[C_u \setminus (S \cup K)])$ as well as Lemma~\ref{lemma:2b}, case~(b), if $T_{u,2}(H) = T_{v,1}(G[C_v])$ and the root node of $T_{v,1}(G[C_v])$ is a $1$-node are not possible either. If Lemma~\ref{lemma:2b},case~(b) holds with $T_{u,2}(H) = T_{v,1}(G[C_v])$ and the root node of $T_{v,1}(G[C_v])$ being a $2$-node then by exchanging $T_{u,1}(H)$ and $T_{u,2}(H)$, we get an optimal solution~$H'$ of the ($P_4$-sparse-2CC,{$+$}$1$)-Min\-Edge\-Addition Problem for the graph~$G$ and the non-edge~$uv$ with $T_{v,1}(H') = T_{v,1}(G[C_v])$.
				
				On the other hand, if $u,v$ in $S_2 \cup K_2$, then Lemma~\ref{lemma:2c} applies. Since $|C_u \setminus (S \cup K)| \ge 3$, the subgraph~$G[C_u \setminus (S \cup K)]$ cannot be a $P_2$; then $G[C_u \setminus (S \cup K)]$ is a headless thin spider (which includes the $P_4$) with spider partition $(S_u,K_u,\emptyset)$ and due to Observation~\ref{obs:2a}, the subgraph~$G[C_u] =(S_G,K_G,R_G)$ is such that $S_G = S \cup S_u$, $K_G = K \cup K_u$, and $R_G = \emptyset$. Then, if $G[C_v]$ is a $P_2$, the number of fill edges (including $uv$) in $H$ is $|K| \cdot (|K_1| + 2) + |K_1| > 2\,(|K| + |K_1|)$, whereas if each of $G[C_u], G[C_v]$ is a $P_4$, the number of fill edges (including $uv$) in $H$ is $6\,|K| + 4 > 2\,(|K| + 2)$; in each case, the number of fill edges in $H$ is greater than the fill edges needed if we make $v$ universal in $G$, a contradiction.
                Thus, the only possibility is Lemma~\ref{lemma:2c}, case~(a), i.e, there exists an optimal solution~$F$ of the ($P_4$-sparse-2CC,{$+$}$1$)-Min\-Edge\-Addition Problem for the graph~$G[V(G) \setminus (S \cup K)]$ and the non-edge~$uv$ in which $u$ or $v$ is universal.
				\item
				\emph{The spider~$(S_2,K_2,R_2)$ is thick}. 
				Since at least one of $u,v$ belongs to $S_2 \cup K_2$, Lemma~\ref{lemma:3bc} implies that there exists an optimal solution~$H'$ with $u$ or $v$ being universal in $H'[V(G) \setminus (S \cup K)]$.
			\end{itemize}
			Therefore, if the least common ancestor~$t$ of $u,v$ is a child of the root of $T_H$ and $t$ is a $2$-node, then there exists an optimal solution~$H'$ of the ($P_4$-sparse-2CC,{$+$}$1$)-Min\-Edge\-Addition Problem for the graph~$G$ and the non-edge~$uv$ in which $T_{u,1}(H') = T_{v,1}(G[C_v])$ or $u$ or $v$ is universal in the induced subgraph~$G[V(G) \setminus (S \cup K)]$.
		\end{itemize}
		If vertex~$u$ or $v$ is universal in an optimal solution of the ($P_4$-sparse-2CC,{$+$}$1$)-Min\-Edge\-Addition Problem for the graph~$G[V(G) \setminus (S \cup K)$ and the non-edge~$uv$, then there exists an optimal solution of the ($P_4$-sparse-2CC,{$+$}$1$)-Min\-Edge\-Addition Problem for $G$ and $uv$ by making $u$ or $v$ universal in $G$; note that $H$ contains the fill edges connecting the vertices in $K$ to the vertices in $(S_G \setminus S) \cup \{v\}$ with $|S_G \setminus S| \ge 1$.
		\item [C.] \emph{The root node of the subtree~$T_{G[C_u]}$ is a $2$-node corresponding to a thick spider~$Q_G = (S_G,K_G,R_G)$}. Since $Q_G$ is a thick spider and $|S_G| = |K_G| \ge 3$, every vertex $w \in C_u$ is adjacent to at least $2$ vertices in $C_u$. On the other hand, in $H$, each vertex in $S \subset C_u$ is adjacent to exactly $1$ vertex, which belongs to $K \subset C_u$. Therefore, such a case is \emph{impossible}. 
	\end{enumerate}
\end{proof}

In addition to the above case, it is possible that $S \cup K$ contains vertices from both $C_u$ and $C_v$.

\begin{lemma} \label{lemma:2a_2}
	Suppose that the optimal solution~$H$ of the ($P_4$-sparse-2CC,{$+$}$1$)-Min\-Edge\-Addition Problem for a $P_4$-sparse graph~$G$ and a non-edge~$u v$ is a thin spider~($S,K,R$) with $u,v \in R$. Then, if $S \cup K$ contains vertices from both $C_u$ and $C_v$, there exists an optimal solution~$H'$ of the ($P_4$-sparse-2CC,{$+$}$1$)-Min\-Edge\-Addition Problem for the graph~$G$ and the non-edge~$uv$
 \begin{itemize}
\item[(i)] which results from making $u$ or $v$ universal in $G$ or
\item[(ii)] in which the subtree~$T_{u,1}(H')$ is identical to $T_{u,1}(G[C_u])$ or $T_{v,1}(G[C_v])$ in the case shown in Figure~\ref{fig:trees}(c), i.e, $u'$ or $v'$ become universal in $G$.
\end{itemize} 
\end{lemma}

\begin{proof}
	Because $H[C_u]$ and $H[C_v]$ are connected (Observation~\ref{obs: solution_prop}(i)), there exist vertices $u' \in K \cap C_u$ and $v' \in K \cap C_v$. Let $k_u = |K \cap C_u|$ and $k_v = |K \cap C_v|$; clearly $k_u \ge 1$ and $k_v \ge 1$. By taking into account the fill edges with one endpoint in $C_u$ and the other in $C_v$, 
    %we have that
    the number~$N$ of fill edges (including $uv$) in $H$ is
	$N \ \ge\  k_u \cdot |R \cap C_v| + k_v \cdot |R \cap C_u| + k_u \cdot k_v + 1$
	where the term~{$+$}$1$ accounts for the
    %added
    non-edge~$uv$. Then by Equation~\ref{eq:product} we have
	$N \ \ge\  (k_u + |R \cap C_v| - 1) + (k_v + |R \cap C_u| - 1) + (k_u + k_v - 1) + 1 \ =\  |V(G)| - 2$.
	If vertex~$u$ is universal in $G$ then the number of fill edges is $|V(G)| - 1 - deg_G (u)$ where $deg_G (u) \ge 1$ and similarly for $v$. Then, the optimality of $H$ and Equation~\ref{eq:product} imply that in $H$ all of the following hold:\\
	\phantom{xxxx} $deg_G (u) = deg_G (v) = 1$;\\
	\phantom{xxxx} $k_u = 1$ or $|R \cap C_v| = 1$;\\
	\phantom{xxxx} $k_v = 1$ or $|R \cap C_u| = 1$;\\
	\phantom{xxxx} $k_u = 1$ or $k_v = 1$;\\
	\phantom{xxxx} no fill edges exist with both endpoints in $C_u$ or $C_v$, i.e., $H[C_u] = G[C_u]$\\
    \phantom{xxxx} and $H[C_v] = G[C_v]$.
	
	The facts that $u \in R$, $deg_G (u) = 1$, $H[C_u] = G[C_u]$, and $k_u \ge 1$ imply that $k_u = 1$ and that in $G$, $u$ is adjacent only to $u' = K \cap C_u$ which is universal in $G[C_u]$ (Figure~\ref{fig:trees}(c)). Similarly, $k_v = 1$ and in $G$, $v$ is adjacent only to $v' = K \cap C_v$, which is universal in $G[C_v]$. Then, $|K| = 2$ and the number of fill edges (including $uv$) in $H$ is $|R| + 2 = |V(G)| - 2$ (where $|R| \ge 2$) matching the number of fill edges if $u$ or $v$ is made universal in $G$.
	
	Moreover, we can get a $P_4$-sparse graph by making $u'$ or $v'$ universal. In particular, if $|R \cap C_v| \le |R \cap C_u|$, we make $u'$ universal and add the fill edge~$u v$, and if $|R \cap C_v| > 1$ we add the fill edge~$v' u$ as well; the total number of fill edges is $4$ if $|R \cap C_v| = 1$ and $|R \cap C_v| + 4$ if $|R \cap C_v| \ge 2$; a symmetric result holds if $|R \cap C_u| \le |R \cap C_v|$. In summary, the number of fill edges (including $uv$) in $H$ is $4$ if $\min\{|R \cap C_u|, \,|R \cap C_v| \} = 1$ otherwise it is $\min\{|R \cap C_u|, \,|R \cap C_v| \} + 4$. Since $\min\{|R \cap C_u|, \,|R \cap C_v| \} \le |R|/2$, this solution ties the solution with $u$ or $v$ universal if $R = \{u,v\}$ or $|R \cap C_u| = |R \cap C_v| = 2$ and is better in all other cases. The lemma follows from the fact that $u'$ ($v'$ respectively) is universal in $G[C_u]$ ($G[C_v]$ respectively).
\end{proof}

\noindent
\emph{Cases 2b. One of the vertices $u,v$ belongs to $R$ and the other one belongs to $S \cup K$}.
\ Since $u,v$ are adjacent in the solution~$H$, the latter vertex belongs to $K$.
Without loss of generality, suppose that $u \in K$ and $v \in R$. Then $k_u = |K \cap C_u| \ge 1$.

\begin{lemma} \label{lemma:2b}
Suppose that an optimal solution~$H$ of the ($P_4$-sparse-2CC,{$+$}$1$)-Min\-Edge\-Addition Problem for a $P_4$-sparse graph~$G$ and a non-edge~$uv$ is a thin spider~($S,K,R$) with one of $u,v$ belonging to $R$ and the other one belonging to $S \cup K$.
Then, there exists an optimal solution~$H'$ which
\begin{itemize}
	\item[(a)]
	results from making $u$ or $v$ universal in $G$ \quad or
	\item[(b)]
	has $T_{u,1}(H') = T_{u,1}(G[C_u])$ or $T_{v,1}(H') = T_{v,1}(G[C_v])$
	\item[(c)]
	except if in $G$, (i)~$C_u$ induces a $P_2$~$u u'$ and $C_v$ induces a $P_3$~$v v' v''$ (in which case, we get a thin spider $(S',K',R')$ with $S' = \{u', v''\}$, $K' = \{u, v'\}$, and $R' = \{v\}$ for a total of $2$ fill edges including $uv$) or $C_v$ induces a thin spider $(S_1,K_1,R_1)$ with $v$ being an isolated vertex in $G[R_1]$ and $|R_1| \le |K_1|$ (in which case we get a thin spider $(S',K',R')$ with $S' = \{u'\} \cup S_1$, $K' = \{u\} \cup K_1$, and $R' = R_1$ for a total of $|K_1| + |R_1|$ fill edges including $uv$) or (ii)~we have a situation similar to (i) with the roles of $C_u, C_v$ interchanged.
	\end{itemize}
\end{lemma}

\begin{proof}
We distinguish the following cases:
	
\begin{itemize}
	\item[A.] $k_v = |K \cap C_v| = 0$.
	Then because $G[C_v]$ is connected, $C_v \subseteq R$. By taking into account the fill edges with one endpoint in $C_u$ and the other in $C_v$, we have that the number~$N$ of fill edges in $H$ (including $uv$) is 
	$N \ \ge\  k_u \cdot |R \cap C_v| \ =\  k_u \cdot |C_v| \ \ge\ k_u + |C_v| - 1$. If we make $u$ universal in $G$, the number of fill edges (including $u v$) is precisely $k_u - 1 + |C_v|$. 
	Then, the optimality of $H$ implies that $N = k_u - 1 + |C_v|$, that is, $k_u \cdot |C_v| = k_u + |C_v| - 1$; then, Equation~\ref{eq:product} requires that  $k_u = 1$ or $|C_v| = 1$, and that no additional fill edges exist, i.e., $G[C_u] = H[C_u]$ and $G[C_v] = H[C_v]$.
    Since $G[C_u]$ is a thin spider, $k_u\ge2$, so $|C_v|=1$ and $N=k_u$. If a solution used fewer than $k_u$ fill edges, there would be a leg $ab$ of $G[C_u]$, with $a\in K\setminus\{u\}$, such that no fill edge is incident with $a$ or $b$. Then $\{u,v,a,b,c\}$, where $c$ is the unique neighbor of $u$ in $S$, induces $F_5$ or $F_2$ according as $vc$ is present or absent, a contradiction.
    % Thus, since $k_v = 0$, $G[C_u]$ is a thin spider ($S_u,K_u,R_u$), which implies that $k_u \ge 2$; thus $|C_v| = 1$, i.e., $C_v = \{v\}$. Then, $N = k_u + |C_v| - 1 = k_u$ and this is optimal: if there were an optimal solution~$H'$ with at most $k_u - 1$ fill edges (one of which is $uv$), there would exist an edge~$ab$ in $G[C_u]$ where $a \in K \setminus \{u\}$, $b \in S$, and no fill edge in $H'$ is incident on $a$ or $b$; then, the vertices $u,v,a,b,c$ (where $c$ is the unique neighbor of $u$ in $S$) induce an $F_5$ or an $F_2$ depending on whether $v,c$ are adjacent in $H'$ or not, in contradiction to the fact that $H'$ is $P_4$-sparse.
	%
	\item[B.] $k_v = |K \cap C_v| \ge 1$ and $R \cap C_u \ne \emptyset$.
	By taking into account the number of fill edges with one endpoint in $C_u$ and the other in $C_v$, we have that the number~$N$ of fill edges (including $uv$) in $H$ is
	\begin{align*}
	N \ \ge\  & k_u \cdot k_v + k_u \cdot |R \cap C_v| + k_v \cdot |R \cap C_u|\\
	\ge\  &(k_u + k_v - 1) + (k_u + |R \cap C_v| - 1) + (k_v + |R \cap C_u| - 1)\\ =\  &|V(G)| - 3.
	\end{align*}
	If we make $u$ universal in $G$, the number of fill edges (including $u v$) is precisely $|V(G)| - 1 - deg_G(u)$. By Observation~\ref{obs:2a} and the facts that the induced graph~$G[C_u]$ is connected (Observation~\ref{obs: solution_prop}(i)) and that $R \cap C_u \ne \emptyset$, we have $deg_G(u) \ge k_u + 1 \ge 2$. Then, the optimality of $H$ implies that $deg_G(u) = 2$ and $N = |V(G)| - 3$ which by Equation~\ref{eq:product} requires that all of the following hold:
		$k_u = 1$ or $k_v = 1$;
		$k_u = 1$ or $|R \cap C_v| = 1$;
		$k_v = 1$ or $|R \cap C_u| = 1$;
		no additional fill edges exist, i.e., $G[C_u] = H[C_u]$ and $G[C_v] = H[C_v]$. Since $deg_G(u) = 2$, $k_u \ge 1$, $R \cap C_u \ne \emptyset$, and $G[C_u] = H[C_u]$, Observation~\ref{obs:2a} implies that $k_u = 1$ and $|R \cap C_u| = 1$; thus, $G[C_u]$ is a $P_3$ and $N = 2\,k_v + |R \cap C_v|$.
		
	Next, if we make $v$ universal in $G$, the number of fill edges (including $u v$) is precisely $3 + k_v + |(R \cap C_v) \setminus N_G[v]|$. The optimality of $H$ implies that
		
		$2\,k_v + |R \cap C_v| \le 3 + k_v + |(R \cap C_v) \setminus N_G[v]|
		\ \Longleftrightarrow\  k_v + |(R \cap C_v) \cap N_G(v)| \le 2$.
		
	Then there exist three possibilities:
		\begin{itemize}
			\item[(i)] $k_v = 1$ and $|(R \cap C_v) \cap N_G(v)| = 0$. Let $K \cap C_v = \{a\}$. If $|R \cap C_v| = 1$, then an optimal solution requires $3$ fill edges (including $uv$), a \emph{tie} between the thin spider~$H$ (clique $\{u,a\}$) and making $u$ universal in $G$;  if $|R \cap C_v| = 2$, then an optimal solution requires $4$ fill edges (including $uv$), a three-way \emph{tie} among the thin spider~$H$, making $u$ universal in $G$, and making $a$ universal in $G$; if $|R \cap C_v| \ge 3$, the optimal solution is obtained by making $a$ universal in $G$, which requires $4$ fill edges (including $uv$). Note that vertex~$a$ is universal in $G[C_v]$; thus, $T_{v,1}(H') = T_{v,1}(G[C_v])$ if $H'$ is the optimal solution with $a$ universal in $G$.
			\item[(ii)] $k_v = 1$ and $|(R \cap C_v) \cap N_G(v)| = 1$. Then $|R \cap C_v| \ge 2$. Let $K \cap C_v = \{a\}$. If $|R \cap C_v| = 2$, then an optimal solution requires $4$ fill edges (including $uv$), a four-way \emph{tie} among the thin spider~$H$ (clique $\{u,a\}$) and making $u$ or $v$ or $a$ universal in $G$; if $|R \cap C_v| \ge 3$, then the optimal solution is to make $a$ universal in $G$ which requires $4$ fill edges (including $uv$). Again, vertex~$a$ is universal in $G[C_v]$ and $T_{v,1}(H') = T_{v,1}(G[C_v])$ if $H'$ is the optimal solution with $a$ universal.
			\item[(iii)] $k_v = 2$ and $|(R \cap C_v) \cap N_G(v)| = 0$. Let $K \cap C_v = \{a,b\}$.
			If $|R \cap C_v| = 1$ ($|R \cap C_v| = 2$ resp.), then an optimal solution requires $5$ ($6$ resp.) fill edges (including $uv$), a three-way \emph{tie} among the thin spider~$H$ (clique $\{u,a,b\}$) and making $u$ or $v$ universal in $G$; if $|R \cap C_v| = 3$, then an optimal solution requires $7$ fill edges (including $uv$), a four-way \emph{tie} among the thin spider~$H$, making $u$ or $v$ universal in $G$, and forming a thin spider with clique $\{a,b\}$; if $|R \cap C_v| \ge 4$, then an optimal solution is to form a thin spider with clique $\{a,b\}$ which requires $7$ fill edges (including $uv$). Again, note that $G[C_v]$ is a thin spider with clique $\{a,b\}$.
		\end{itemize}
		\item[C.] $k_v \ge 1$ and $R\cap C_u = \emptyset$. Then $R \cap C_v = R$.
		By taking into account the number of fill edges with one endpoint in $C_u$ and the other in $C_v$, we have that the number~$N$ of fill edges (including $uv$) in $H$ is
		\begin{align*}
		N \ \ge\ \  & k_u \cdot k_v + k_u \cdot |R \cap C_v| \ =\  k_u \cdot k_v + k_u \cdot |R| \\   %\nonumber
		\ge\ \  & (k_u + k_v - 1) + (k_u + |R| - 1) \ =\  2\,k_u + k_v + |R| - 2.
		\end{align*}
		In accordance with Observation~\ref{obs:2a}, if we make $u$ universal in $G$ then the number of fill edges is $k_u - 1 + 2\, k_v + |R|$ whereas if we make $v$ universal in $G$ then the number of fill edges is $2\, k_u + k_v + |R \setminus N_G[v]|$.
		The optimality of $H$ implies that
		
		$2\,k_u + k_v + |R \setminus N_G[v]| \ \ge \  N \ \ge\  2\,k_u + k_v + |R| - 2
		\ \Longleftrightarrow\   |R \cap N_G(v)| \le 1$
		
		and in accordance with Equation~\ref{eq:product} for the product $k_u \cdot |R|$, that
		
		$k_u - 1 + 2\, k_v + |R| \ \ge \  N \ \ge\  k_u \cdot k_v + k_u \cdot |R| \ \ge\  k_u \cdot k_v + k_u + |R| - 1$
		
		from which we conclude that $k_u \le 2$.
		In fact, if $k_u = 2$, then from $k_u - 1 + 2\, k_v + |R| \ge k_u \cdot k_v + k_u \cdot |R|$ we conclude that
		$2\, k_v + |R| + 1 \ge 2\, k_v + 2\, |R|
		\ \Longleftrightarrow\  |R| + 1 \ge 2\, |R|
		\ \Longleftrightarrow\  |R| \le 1
		\ \Longleftrightarrow\  |R| = 1$, i.e., $R = \{v\}$.
		
		We distinguish two cases.
		\begin{itemize}
			\item[(i)]
			\emph{$v$ has no neighbors in $R$}.
			We have that $k_u = 1$ or $k_u = 2$.\\
            \textbf{$\boldsymbol{k_u = 1}$.}
            Then $G[C_u]$ is a $P_2$.
			If $k_v = 1$ then if $|R| = 1$ (i.e., $G[C_v]$ is a $P_3$), the optimal solution is the thin spider~$H$ which requires $2$ fill edges (including $uv$), if $|R| \ge 3$ the optimal solution is to make the single vertex in $K \cap C_v$ universal which requires $3$ fill edges (including $uv$), and there is a \emph{tie} between these two possibilities if $|R| = 2$ ($3$ fill edges including $uv$); note that the single vertex in $K \cap C_v$ is universal in $G[C_v]$.
   			Let us now consider that $k_v \ge 2$. We note that in this case making $u$ universal in $G$ requires more fill edges than making $v$ universal in $G$ which in turn ties the thin spider~$H$ if $|R \cap N_G(v)| = 1$ and requires more fill edges if $|R \cap N_G(v)| = 0$. Then, if $|R| \le k_v$, the optimal solution is the thin spider~$H$ which requires $|R| + k_v$ fill edges (including $uv$), if $|R| \ge k_v + 2$ the optimal solution is the thin spider~$H'$ with clique $K \cap C_v$ (the vertices in $C_u$ are placed in the $R$-set of the spider) which requires $2\, k_v + 1$ fill edges (including $uv$), and there is a \emph{tie} between these two possibilities if $|R| = k_v + 1$ in which case $|R| + k_v = 2\, k_v + 1$ fill edges (including $uv$) are required. Note that for the spider~$H'$ it holds that $T_{v,1}(H') = T_{v,1}(G[C_v])$.\\
			%\textbf{(OPTIMAL?)}
			%            
            \textbf{$\boldsymbol{k_u = 2}$.} 
            Then $G[C_u]$ is a $P_4$. Moreover, $R = \{v\}$ and $G[C_v]$ is a $P_3$ if $k_v = 1$ or else $G[C_v]$ is a thin spider $(S_v,K_v,R_v)$ where $|S_v| = |K_v| = k_v \ge 2$ and $R_v = \{v\}$.
			If $G[C_v]$ is a $P_3$ then an optimal solution requires $4$ fill edges (including $uv$), a \emph{tie} between the thin spider~$H$ and making $u$ universal in $G$; if $G[C_v]$ is a thin spider $(S_v,K_v,\{v\})$, then if $k_v = 2$ an optimal solution requires $6$ fill edges (including $uv$), a \emph{tie} between the thin spider~$H$ and making $u$ or $v$ universal in $G$ whereas if $k_v \ge 3$, the optimal solution is to make $v$ universal which requires $k_v + 4$ fill edges (including $uv$).
            
			%\textbf{(OPTIMAL?)}
			%
			\item[(ii)]
			\emph{$v$ has $1$ neighbor in $R$}.
			Let $z$ be the neighbor of $v$ in $R$ and let $S_z$ be the connected component in $H[R]$ to which $v,z$ belong. The fact that $v$ has $1$ neighbor in $R$ implies that $|R| = |R \cap C_v| \ge 2$ and hence $k_u = 1$; then, due to Observation~\ref{obs:2a}, the induced subgraph~$G[C_u]$ is a $P_2$.
            Moreover, $|R \setminus N_G[v]| = |R| - 2$ and the optimality of $H$ implies that $N =  2\,k_u + k_v + |R| - 2 = k_v + |R|$ which by Equation~\ref{eq:product} requires that no additional fill edges exist, i.e., $H[C_u] = G[C_u]$ and $H[C_v] = G[C_v]$. Then, the thin spider~$H$ and making $v$ universal in $G$ \emph{tie} in the number of fill edges required.
			If $k_v = 1$, the optimal solution~$H'$ is making the single vertex in $K \cap C_v$ universal in $G$ which requires $3$ fill edges if $S_z = \{v,z\}$, $4$ fill edges if $G[S_z]$ is a $P_3$, $5$ fill edges if $z$ is universal in $S_z$ but $S_z$ is not a $P_2$ or a $P_3$, $\kappa + 4$ fill edges if $zv$ is a "leg" of a thin spider with clique size equal to $\kappa$.
   %          \par
			% \phantom{xx} $3$ fill edges if $S_z = \{v,z\}$,
			% \par
			% \phantom{xx} $4$ fill edges if $G[S_z]$ is a $P_3$,
			% \par
			% \phantom{xx} $5$ fill edges if $z$ is universal in $S_z$ but $S_z$ is not a $P_2$ or a $P_3$,
			% \par
			% \phantom{xx} $\kappa + 4$ fill edges if $zv$ is a "leg" of a thin spider with clique size
			% \par
			% \phantom{xxxxxx} equal to $\kappa$
			\par
			If $k_v \ge 2$ then the induced subgraph~$G[C_v]$ is a thin spider $(S_v,K_v,R)$ with $K_v = K \cap C_v$. Then, a thin spider $(S_v,K_v,R_v \cup C_u)$ can be built which requires $2\,k_v + 1$ fill edges if $S_z = \{v,z\}$, $2\,k_v + 2$ fill edges if $G[S_z]$ is a $P_3$, $2\,k_v + 3$ fill edges if $z$ is universal in $S_z$ but $S_z$ is not a $P_2$ or a $P_3$, $2\,k_v + \kappa + 2$ fill edges if $zv$ is a "leg" of a thin spider with clique size equal to $\kappa$, 
   %          \par
			% \phantom{xx} $2\,k_v + 1$ fill edges if $S_z = \{v,z\}$,
			% \par
			% \phantom{xx} $2\,k_v + 2$ fill edges if $G[S_z]$ is a $P_3$,
			% \par
			% \phantom{xx} $2\,k_v + 3$ fill edges if $z$ is universal in $S_z$ but $S_z$ is not a $P_2$ or a $P_3$,
			% \par
			% \phantom{xx} $2\,k_v + \kappa + 2$ fill edges if $zv$ is a "leg" of a thin spider with clique size
			% \par
			% \phantom{xxxxxx} equal to $\kappa$
			% \par
			where the above number of fill edges includes $uv$. For both $k_v = 1$ and $k_v \ge 2$, it holds that for the optimal solution~$H'$ we have that $T_{v,1}(H') = T_{v,1}(G[C_v])$.
			%The optimal solution is one of the above possibilities and depends on the difference of $|R| - kv$. 
			%\textbf{(OPTIMAL?)}
		\end{itemize}
	\end{itemize}
\end{proof}

%\newpage
\noindent
\emph{Case 2c. The vertices $u,v$ belong to $S \cup K$.}
Since $uv \in E(H)$ and $S$ is an independent set, exactly one of the following two possibilities occurs.

\paragraph{Case 2c.1.} $u, v \in K$.

In this case, let
$
k_u = |K \cap C_u|$ and $k_v = |K \cap C_v|$. Clearly,
$
k_u \geq 1$, $ k_v \geq 1.
$

\paragraph{Case 2c.2.} The prescribed edge $uv$ is a leg of the thin spider $H$.

Without loss of generality, assume that
$
v \in S$, $u \in K$. The two cases are considered separately below.

%Since $u,v$ are adjacent in $H$ and because of Observation~\ref{obs:2a}(i), then $u, v \in K$ and thus $k_u = |K \cap C_u| \ge 1$ and $k_v = |K \cap C_v| \ge 1$. 
%We show the following lemma.

\begingroup
\renewcommand{\thelemma}{\arabic{lemma}A}
\begin{lemma} \label{lemma:2c}
Suppose that an optimal solution $H$ of the $(P_4\text{-sparse-2CC},+1)$-MinEdgeAddition Problem for the graph $G$ and the non-edge $uv$ is a thin spider $(S,K,R)$ with
$
u,v \in K.
$
Then there exists an optimal solution which
\begin{enumerate}[(a)]
    \item results from making either $u$ or $v$ universal in $G$,
\end{enumerate}
except if in $G$
\begin{enumerate}[(i)]
    \item one of $C_u, C_v$ induces a $P_2$ and the other induces either a $P_2$ or a headless thin spider $(S_1, K_1, \emptyset)$, with the corresponding endpoint $u$ or $v$ belonging to $K_1$; or
    \item both $C_u$ and $C_v$ induce a $P_4$, with $u$ and $v$ being internal vertices of these paths,
\end{enumerate}
in which cases an optimal solution is obtained by joining $G[C_u]$ and $G[C_v]$ into a thin spider.
\end{lemma}
\endgroup
\begin{proof}
	Let	$k_u=|K\cap C_u|$ and $k_v=|K\cap C_v|$. Recall that in the present case $u,v\in K$. We first record two consequences of Observation~\ref{obs:2a} that will be used throughout the proof. Since neither endpoint of the prescribed non-edge belongs to $S$, every leg of the thin spider $H$ belongs to $E(G)$. Consequently, the two endpoints of every leg belong to the same connected component of $G$. Thus,	$|S\cap C_u|=|K\cap C_u|=k_u$ and $|S\cap C_v|=|K\cap C_v|=k_v$. Moreover, by Observation~\ref{obs:2a}(ii), every edge of $H[K\cap C_u]$ and every edge of $H[K\cap C_v]$ already belongs to $G$. We will also use the following claim.

	\noindent
	{\bf Claim.}
	If $x\in\{u,v\}$ and $R\cap C_x\ne\emptyset$, then every vertex of
	$K\cap C_x$ has a neighbor in $R\cap C_x$ in $G$. In particular, $\deg_G(x)\ge k_x+1$.
	
	\smallskip
	\noindent
	{\it Proof of the claim.}
	Since $G[C_x]$ is connected, $R\cap C_x\ne\emptyset$, and vertices of $R$ have no neighbors in $S$ in $H$ (and hence neither in $G$), there exists an edge $ar\in E(G)$ with $a\in K\cap C_x$ and $r\in R\cap C_x$. Let $b\in K\cap C_x$. If $b=a$, then $br\in E(G)$. Suppose that $b\ne a$ and, for contradiction, that $br\notin E(G)$. Let $s_a,s_b\in S\cap C_x$ be the vertices matched to $a,b$, respectively. As noted above, the legs $s_aa$ and $s_bb$ as well as the edge $ab$ belong to $E(G)$. Then the five vertices $s_b,b,a,s_a,r$ contain the two induced $P_4$'s $s_b-b-a-s_a$ and $s_b-b-a-r$. This contradicts the fact that $G$ is $P_4$-sparse. Hence
	$br\in E(G)$.
	
	Therefore every vertex of $K\cap C_x$ is adjacent in $G$ to $r$. Together with its matched vertex in $S\cap C_x$ and the other $k_x-1$ vertices of $K\cap C_x$, this yields $\deg_G(x)\ge (k_x-1)+1+1=k_x+1$.	This proves the claim.
	\hfill$\diamond$
	
	\smallskip
	
	Due to the symmetry of $u,v$, it suffices to consider the following
	cases.
	\begin{itemize}
		\item[A.]
		$R\cap C_u\ne\emptyset$ and $R\cap C_v\ne\emptyset$.
		
		% By counting the fill edges with one endpoint in $C_u$ and the
		% other in $C_v$, we obtain $N\ge
		% k_u k_v+
		% k_u|R\cap C_v|+
		% k_v|R\cap C_u|$. By Equation~\ref{eq:product},
		% \[
		% \begin{split}
		% N
		% &\ge
		% (k_u+k_v-1)
		% +(k_u+|R\cap C_v|-1)
		% +(k_v+|R\cap C_u|-1)\\
		% &=|V(G)|-3.
		% \end{split}
		% \]
		
		% On the other hand, making $u$ universal in $G$ gives a feasible	solution using $|V(G)|-1-\deg_G(u)$ fill edges, including $uv$. Since $H$ is optimal, 
		% $N\le |V(G)|-1-\deg_G(u)$. By the claim, $\deg_G(u)\ge k_u+1\ge2$, and hence $N\le |V(G)|-3$. 
        Counting cross-component fill edges and applying Equation~1 gives $N\ge |V(G)|-3$. Since $\deg_G(u)\ge k_u+1\ge2$ by the Claim, making $u$ universal costs at most $|V(G)|-3$. Hence optimality gives $N=|V(G)|-3$ and $\deg_G(u)=2$.
        %Combining this inequality with the preceding lower bound gives $N=|V(G)|-3$ and $\deg_G(u)=2$.	
        Since $\deg_G(u)\ge k_u+1$, we obtain $k_u=1$.
		
		By symmetry, $k_v=1$ and $\deg_G(v)=2$. All inequalities used above are therefore equalities. In particular, all fill edges of $H$ are precisely the cross-component fill edges counted above; hence $G[C_u]=H[C_u]$ and $G[C_v]=H[C_v]$. In $H[C_u]$, the vertex $u$ is adjacent to its matched vertex in $S\cap C_u$ and to every vertex of $R\cap C_u$. Thus $\deg_G(u) =
		1+|R\cap C_u|$.	Since $\deg_G(u)=2$, it follows that $|R\cap C_u|=1$. Similarly, $|R\cap C_v|=1$.	Consequently, both $G[C_u]$ and $G[C_v]$ are $P_3$'s.
		
		The thin spider $H$ requires $N=3$ fill edges, including $uv$, and this ties the solution obtained by making either $u$ or $v$ universal in $G$.

		\item[B.]
		$R\cap C_u\ne\emptyset$ but $R\cap C_v=\emptyset$.
		
		Then $R\cap C_u=R$. By counting the fill edges with one endpoint in $C_u$ and the other in $C_v$, we obtain $N\ge k_u k_v+k_v|R|$. Using Equation~\ref{eq:product},
		\[
		\begin{split}
		N
		&\ge
		(k_u+k_v-1)+(k_v+|R|-1)\\
		&=|V(G)|-k_u-2.
		\end{split}
		\]
		
		By the claim, $\deg_G(u)\ge k_u+1$. Therefore, making $u$ universal requires at most $|V(G)|-k_u-2$ fill edges. Since $H$ is optimal, $|V(G)|-k_u-2
		\le N
		\le |V(G)|-1-\deg_G(u)
		\le |V(G)|-k_u-2$. Hence equality holds throughout:	$N=|V(G)|-k_u-2$ and $\deg_G(u)=k_u+1$.
		
		Equality in the two applications of Equation~\ref{eq:product} implies 
		$k_u=1\ \textrm{or}\ k_v=1$ 
		and 
		$k_v=1\ \textrm{or}\ |R|=1$. Furthermore, equality in the lower bound for $N$ implies that there are no additional fill edges inside either component.
		Consequently, $G[C_u]=H[C_u]$ and $G[C_v]=H[C_v]$.
		
		Since $G[C_u]=H[C_u]$, the degree of $u$ is 
		$\deg_G(u)=
		(k_u-1)+1+|R|
		=
		k_u+|R|$.
		Together with $\deg_G(u)=k_u+1$, this gives	$|R|=1.$
		
		Since $R\cap C_v=\emptyset$ and $G[C_v]=H[C_v]$, we also have $\deg_G(v)=k_v$. Moreover, comparing $H$ with the solution obtained by making $v$
		universal gives $N\le |V(G)|-1-\deg_G(v)$. As $N=|V(G)|-k_u-2$, it follows that
		$\deg_G(v)\le k_u+1$, and therefore	$k_v\le k_u+1$.	Since at least one of $k_u,k_v$ is equal to $1$, we distinguish
		the following subcases.
		
		\begin{itemize}
			\item
			$k_u=k_v=1$.Then $G[C_u]$ is a $P_3$ and $G[C_v]$ is a $P_2$. The thin spider $H$ requires $2$ fill edges, including $uv$, and ties the solution obtained by making $u$ universal in $G$.
			
			\item
			$k_u=1$ and $k_v>1$. Since $k_v\le k_u+1=2$, we have $k_v=2$. Thus $G[C_u]$ is a $P_3$ and $G[C_v]$ is a $P_4$. The thin spider $H$ requires $4$ fill edges, including $uv$, and ties the solutions obtained by making $u$ or $v$ universal in $G$.
			
			\item
			$k_v=1$ and $k_u>1$. Then $G[C_v]$ is a $P_2$, whereas $G[C_u]$ is a thin spider with clique size $k_u$ and exactly one vertex in its $R$-set. The thin spider $H$ requires $k_u+1$ fill edges, including $uv$, and ties the solution obtained
			by making $u$ universal in $G$.
		\end{itemize}

		\item[C.]
		$R=\emptyset$.
		
		By Observation~\ref{obs:2a}(i), every leg of $H$ belongs to $E(G)$, and, by Observation~\ref{obs:2a}(ii), all edges between vertices of $K$ that lie in the same component belong to $E(G)$. Since $R=\emptyset$, these are all the edges of $H$ whose endpoints lie in the same component. Therefore $G[C_u]=H[C_u]$ and $G[C_v]=H[C_v]$. In particular, $\deg_G(u)=k_u$ and $\deg_G(v)=k_v$.
		
		The fill edges of $H$ are precisely the edges joining a vertex of $K\cap C_u$ to a vertex of $K\cap C_v$. Hence $N=k_u k_v$. Suppose, without loss of generality, that $k_u\ge k_v$. Making $u$ universal in $G$ requires	$|V(G)|-1-\deg_G(u)
		=
		k_u+2k_v-1$	fill edges. Since $H$ is optimal, $k_u k_v \le k_u+2k_v-1$. Using $k_v\le k_u$, we obtain $k_u k_v
		\le k_u+2k_v-1
		\le 3k_u-1
		<3k_u$,	and therefore $k_v<3$. We distinguish the following cases.
		
		\begin{itemize}
			\item
			$k_v=1$. Then $G[C_v]$ is a $P_2$ and $G[C_u]$ is a headless thin spider; $H$ requires $k_u$ fill edges, whereas making $u$ universal requires $k_u+1$. This gives exception~(i).
			
			\item
			$k_v=2$. Then $G[C_v]$ is a $P_4$, and optimality gives $2k_u\le k_u+3$, hence $k_u\in\{2,3\}$. If $k_u=2$, then $G[C_u]$ is also a $P_4$, giving exception~(ii). If $k_u=3$, the thin-spider solution uses $6$ fill edges and ties the solution obtained by making $u$ universal, so case~(a) holds.
		\end{itemize}
	\end{itemize}
	
	Thus, except for the two configurations stated in the lemma, there
	exists an optimal solution obtained by making $u$ or $v$ universal
	in $G$. In the exceptional configurations, the optimal solution under
	consideration is obtained by joining $G[C_u]$ and $G[C_v]$ into a
	thin spider.
\end{proof}

\begingroup
\addtocounter{lemma}{-1}
\renewcommand{\thelemma}{\arabic{lemma}B}
\begin{lemma} \label{lemma:2cbb}
Suppose that an optimal solution $H$ of the $(P_4\text{-sparse-2CC},+1)$-MinEdgeAddition Problem for the graph $G$ and the non-edge $uv$ is a thin spider $(S,K,R)$, and let $f:S\to K$ be its matching
bijection. Suppose that $uv$ is a leg of $H$. Without loss of generality, let
$
v \in S$, $ u \in K$, $f(v) = u.
$
Then:
\begin{enumerate}[(i)]
    \item 
    $
    C_v = \{v\};
    $
    \item every leg of $H$ different from $uv$ belongs to $E(G)$;
    \item every fill edge of $H$ whose two endpoints belong to $K \cap C_u$ is incident with $u$.
\end{enumerate}
\end{lemma}
\endgroup
\begin{proof}

(i) Since $v\in S$, its only neighbor in $H$ is $u=f(v)$. As $uv\notin E(G)$ and $G[C_v]$ is connected, we get $C_v=\{v\}$.

(ii) follows directly from Observation~2(i).

For (iii), let $ac\in E(H)\setminus E(G)$ with $a,c\in K\cap C_u$. By Observation~2(iii) together with (i), one of $f^{-1}(a),f^{-1}(c)$ is $v$; since $f(v)=u$, we have $a=u$ or $c=u$.
% Since $v \in S$ and $uv$ is its leg, the only neighbor of $v$ in $H$ is $u$. However,
% $
% uv \notin E(G).
% $
% Since $G$ is a spanning subgraph of $H$, the vertex $v$ has no neighbor in $G$. Because $G[C_v]$ is connected, it follows that
% $
% C_v = \{v\}.
% $
% This proves (i).

% Part (ii) follows directly from Observation 2(i), because the stable-set endpoint of every leg different from $uv$ is different from both $u$ and $v$.

% For (iii), let
% $
% ac \in E(H) \setminus E(G)$, $a, c \in K \cap C_u.
% $
% By Observation 2(iii), one of the matched vertices
% $
% f^{-1}(a)$, $f^{-1}(c)
% $
% must be the endpoint outside $C_u$. By part (i), that endpoint is $v$. Since
% $
% f(v) = u,
% $
% we obtain
% $
% a = u$ or $c = u$.
% Therefore every fill edge with both endpoints in $K \cap C_u$ is incident with $u$.
\end{proof}

% \begin{lemma} \label{lemma:2c}
% 	Suppose that an optimal solution~$H$ of the ($P_4$-sparse-2CC,{$+$}$1$)-Min\-Edge\-Addition Problem for a $P_4$-sparse graph~$G$ and a non-edge~$uv$ is a thin spider~($S,K,R$) with $u,v \in S \cup K$.
% 	Then, there exists an optimal solution which
% 	\begin{itemize}
% 		\item[(a)]
% 		results from making either $u$ or $v$ universal in $G$ 
% 		\item[(b)]
% 		except if in $G$\\
% 		(i)~one of $C_u, C_v$ induces a $P_2$ and the other induces a $P_2$ or a headless thin spider  $(S_1,K_1,\emptyset)$ with $u$ or $v$ in $K_1$ or\\
% 		(ii)~both $C_u$ and $C_v$ induce a $P_4$ with $u,v$ being middle vertices,\\
% 		in which cases the optimal solution involves joining $G[C_u]$ and $G[C_v]$ into a thin spider.
% 	\end{itemize}
% \end{lemma}

\begingroup
\addtocounter{lemma}{-1}
\renewcommand{\thelemma}{\arabic{lemma}C}
\begin{lemma}\label{lemma:9c}
	Suppose that an optimal solution $H$ of the 	$(P_4\text{-sparse-2CC},+1)$-MinEdgeAddition Problem for the graph $G$ and the non-edge $uv$ is a thin spider $(S,K,R)$, 	and let $f:S\to K$ be its matching bijection. Suppose that $uv$ is a leg of $H$. Without loss of generality, let $v\in S$, $u\in K$, and $f(v)=u$. Let $N$ denote the number of fill edges in $H$, including $uv$. Then $C_v=\{v\}$ and the following hold.
	
	\begin{itemize}
		\item[(i)]
        If $|K|=2$, then the minimum number of fill edges among such
        leg-completions is
        \[
        \begin{aligned}
        L_2(C_u,u)
        &=
        \min_{\substack{s\in C_u\setminus\{u\}\\
        N_{G[C_u]}(s)=\{k\},\ k\ne u}}
        \Bigl(
        1+
        \bigl|\{z\in C_u\setminus\{s,k\}:kz\notin E(G)\}\bigr|
        \\[-1mm]
        &\hspace{3.8cm}
        +
        \bigl|\{z\in C_u\setminus\{u,s,k\}:uz\notin E(G)\}\bigr|
        \Bigr).
        \end{aligned}
        \]
        where $L_2(C_u,u)=+\infty$ if no such pair $(s,k)$ exists.
		
		\item[(ii)]
		If $|K|\ge3$, then $G[C_u]$ is a thin spider
		$(S_G,K_G,R_G)$. Moreover, exactly one of the following holds:
		\[
		\begin{array}{lll}
		u\in R_G:
		&
		N=1+
		\bigl|\{r\in R_G\setminus\{u\}:ur\notin E(G)\}\bigr|;
		\\[1ex]
		u\in S_G:
		&
		N=|K_G|+|R_G|;
		\\[1ex]
		u\in K_G:
		&
		N=|K_G|.
		\end{array}
		\]
	\end{itemize}
	
	A symmetric statement holds when $u\in S$, $v\in K$, and
	$f(u)=v$.
\end{lemma}
\endgroup

\begin{proof}
	By Lemma~\ref{lemma:2cbb}, since $uv$ is a leg of $H$ and,
	without loss of generality, $v\in S$ and $f(v)=u$, we have $C_v=\{v\}$. Moreover, every leg of $H$ different from $uv$ belongs to $E(G)$.
	
	Let $S_0=S\setminus\{v\}$ and $K_0=K\setminus\{u\}$. Then $|S_0|=|K_0|=|K|-1$. For every $s\in S_0$, the edge $sf(s)$ is a leg of $H$ different from $uv$ and hence, by Lemma~\ref{lemma:2cbb}, $sf(s)\in E(G)$. Since a vertex of the stable set of a thin spider has no other
	neighbor, it follows that
	\[
	N_{G[C_u]}(s)=\{f(s)\}.
	\tag{4}
	\]
	
	We distinguish the cases $|K|=2$ and $|K|\ge3$.

	\noindent
	{\bf Case 1: $|K|=2$.}
	
	Then $S_0=\{s\}$ and $K_0=\{k\}$, where $f(s)=k$. By (4), $N_{G[C_u]}(s)=\{k\}$. Set $X=C_u\setminus\{u,s,k\}$. Consider the graph $\widehat H$ with spider partition
	$\widehat S=\{v,s\}$, 
	$\widehat K=\{u,k\}$,
	$\widehat R=X$,	in which $\widehat H[X]=G[X]$.
	
	To obtain $\widehat H$ from $G$, we need to add $uv$, every missing
	edge joining $k$ to a vertex of $\{u\}\cup X$, and every missing
	edge joining $u$ to a vertex of $X$. No other fill edge is required.
	Indeed, $s$ is adjacent in $G$ only to $k$, $v$ is isolated in
	$G$, and $G[X]$ is $P_4$-sparse because it is an induced subgraph
	of the $P_4$-sparse graph $G$. Thus $\widehat H$ is a
	$P_4$-sparse thin spider. Consequently, the number of fill edges in this completion is precisely $\lambda_2(s,k)
	=
	1+
	\bigl|\{z\in\{u\}\cup X:kz\notin E(G)\}\bigr|
	+
	\bigl|\{z\in X:uz\notin E(G)\}\bigr|$.
	
	The optimal solution $H$ itself determines one admissible pair 	$(s,k)$, and replacing its induced subgraph on $R$ by the corresponding induced subgraph of $G$ cannot require more fill edges. Hence $L_2(C_u,u)\le N$. On the other hand, every pair occurring in the definition of $L_2(C_u,u)$ yields a valid $P_4$-sparse completion of $G+uv$. Since $H$ is optimal, $N\le L_2(C_u,u)$. Therefore $N=L_2(C_u,u)$. This proves part (i).

	\noindent
	{\bf Case 2: $|K|\ge3$.}
	
	Put $q=|S_0|=|K_0|=|K|-1$. Then $q\ge2$. Choose two distinct vertices 
	$s_1,s_2\in S_0$ and let $k_1=f(s_1)$, $k_2=f(s_2)$. By (4), $N_{G[C_u]}(s_i)=\{k_i\}, i=1,2$ and $k_1\ne k_2$. We first show that $G[C_u]$ is a thin spider. The graph $G[C_u]$ is connected by Observation~\ref{obs: solution_prop}(i). Suppose that its complement were disconnected. Then $G[C_u]$ would be the join of at least two nonempty induced subgraphs. Let $Q$ be the part containing $s_1$. Since every vertex of $Q$ is adjacent to every vertex outside $Q$ and $\deg_{G[C_u]}(s_1)=1$,	there can be only one vertex outside $Q$, and that vertex must be $k_1$. Hence $C_u\setminus Q=\{k_1\}$. But $s_2\ne k_1$, so $s_2\in Q$. Since $G[C_u]$ is a join between $Q$ and $C_u\setminus Q$, the edge $s_2k_1$ belongs to $E(G)$.
	This contradicts $N_{G[C_u]}(s_2)=\{k_2\}$ and $k_1\ne k_2$.
	
	Therefore the complement of $G[C_u]$ is connected. Since $G[C_u]$ is a connected $P_4$-sparse graph, the structural characterization of $P_4$-sparse graphs implies that $G[C_u]$ is a spider.
	
	It cannot be a thick spider. According to our convention, a thick spider has clique size at least $3$, and every vertex in its stable set is adjacent to at least two clique vertices. Moreover, vertices in its clique and $R$-set also have degree at least $2$. This is incompatible with the vertices $s_1,s_2$, which have degree $1$ in $G[C_u]$. Thus $G[C_u]$ is a thin spider.
	
	Let $G[C_u]=(S_G,K_G,R_G)$ and let $g:S_G\rightarrow K_G$ be its matching bijection. Every vertex of $S_0$ has degree $1$ in $G[C_u]$. Hence $S_0\subseteq S_G$. Furthermore, the unique neighbor of every $s\in S_0$ is $f(s)$,
	thus $g(s)=f(s)$. Consequently, $K_0\subseteq K_G$. Define $A=S_G\setminus S_0$ and $B=K_G\setminus K_0$. Since $g$ maps $S_0$ bijectively onto $K_0$, it maps $A$	bijectively onto $B$. Thus
	\[
	|A|=|B|=:t.
	\tag{5}
	\]
	
	The vertex $u$ belongs neither to $S_0$ nor to $K_0$. Therefore	$u\in A\cup B\cup R_G$. We consider these three possibilities separately.
	
	\noindent
	{\bf Case 2(a): $u\in R_G$.}
	
	We claim that $t=0$. Suppose, for contradiction, that $t>0$. Construct a graph $F$ with spider partition $S_F=S_G\cup\{v\}$, $K_F=K_G\cup\{u\}$, and $R_F=R_G\setminus\{u\}$, with the matching $g$ on $S_G$ together with the leg $vu$.
	
	The graph $F$ is a completion of $G+uv$. Indeed, since 	$u\in R_G$, the vertex $u$ is already adjacent in $G$ to every 	vertex of $K_G$. Thus the only fill edges required by $F$, apart 	from $uv$, are the missing edges $ur$, $r\in R_G\setminus\{u\}$. The induced graph on $R_F$ is taken to be $G[R_G\setminus\{u\}]$, which is $P_4$-sparse.
	
	On the other hand, every vertex $a\in A$ belongs to the $R$-set of $H$. Since $a\in S_G$ in the graph $G[C_u]$, it is nonadjacent in $G$ to every vertex of $K_0$ and also to $u\in R_G$. Because every vertex of the $R$-set of $H$ must be adjacent to every vertex of $K=K_0\cup\{u\}$, the graph $H$ contains, for each $a\in A$, the $q+1$ fill edges joining $a$ to $K_0\cup\{u\}$.
	
	None of these $(q+1)t$ fill edges is needed in $F$, whereas every
	fill edge required in $F$ is also required in $H$. Hence $|E(F)\setminus E(G)|
	\le
	N-(q+1)t
	<
	N$, contradicting the optimality of $H$. Therefore $t=0$. It follows that $S_G=S_0=S\setminus\{v\}$ and $K_G=K_0=K\setminus\{u\}$. The only fill edges required to obtain the corresponding thin spider are $uv$ and the missing edges between $u$ and $R_G\setminus\{u\}$. Therefore $N=
	1+
	\bigl|
	\{r\in R_G\setminus\{u\}:ur\notin E(G)\}
	\bigr|$. This proves part (ii)(a).

	\noindent
	{\bf Case 2(b): $u\in A$.}
	
	Let	$a=g(u)\in B$. Since $u\in A$, Equation~(5) gives $t\ge1$. We claim that $t=1$. Suppose, for contradiction, that $t>1$, and put $A'=A\setminus\{u\}$, $B'=B\setminus\{a\}$. Then $|A'|=|B'|=t-1$. Construct a graph $F$ with $S_F=S_0\cup A'\cup\{v\}$, $K_F=K_0\cup B'\cup\{u\}$, and $R_F=R_G\cup\{a\}$. The matching consists of the original pairs between $S_0$ and $K_0$, the pairs given by $g$ between $A'$ and $B'$, and the leg $vu$.
	
	This is a valid thin-spider completion. The fill edges that it may 	require are $uv$, the edges joining $u$ to $K_0\cup B'$, and the edges joining $u$ to $R_G$. The vertex $a$ is already adjacent in $G$ to every vertex of $K_0\cup B'$ because $a\in K_G$, and it is adjacent to $u$ because $a=g(u)$.
	
	Now consider a vertex $x\in A'$. In the solution $H$, the vertex $x$ belongs to $R$. In $G[C_u]$, however, $x\in S_G$, and hence its only neighbor in $K_G$ is its matched vertex $g(x)\in B'$. In particular, $x$ is nonadjacent in $G$ to every vertex of $K_0\cup\{u\}$. Since $x\in R$ in $H$, all the $q+1$ edges joining $x$ to $K_0\cup\{u\}$ are fill edges of $H$.
	
	These edges are not required in $F$. All fill edges required by $F$ are already required in $H$: the edges from $u$ to $K_0$ make $u$ a clique vertex in both graphs, the edges from $u$ to $B'$ are required in $H$ because the vertices of $B'$ belong to $R$, and the edges from $u$ to $R_G$ are also required in $H$.
	Thus $|E(F)\setminus E(G)|
	\le
	N-(q+1)(t-1)
	<
	N$, a contradiction. Hence $t=1$. Therefore $A=\{u\}$, $B=\{a\}$, and consequently $S_G=S_0\cup\{u\}
	=
	(S\setminus\{v\})\cup\{u\}$, $
	K_G=K_0\cup\{a\}
	=
	(K\setminus\{u\})\cup\{a\}$.
	
	In the resulting canonical completion, the required fill edges are $uv$, the $|K_G|-1$ edges joining $u$ to $K_G\setminus\{a\}$, and the $|R_G|$ edges joining $u$ to $R_G$. All these edges are absent from $G$ because $u\in S_G$ and
	$a=g(u)$ is its unique clique neighbor. Hence $N
	=
	1+(|K_G|-1)+|R_G|
	=
	|K_G|+|R_G|$. This proves part (ii)(b).

	\noindent
	{\bf Case 2(c): $u\in B$.}
	
	Let $a=g^{-1}(u)\in A$. Again $t\ge1$. We claim that $t=1$. Suppose, for contradiction, that $t>1$, and put $A'=A\setminus\{a\}$, $B'=B\setminus\{u\}$. Then $|A'|=|B'|=t-1$. 
	
	Construct the thin spider $F$ with $S_F=S_0\cup A'\cup\{v\}$, $K_F=K_0\cup B'\cup\{u\}=K_G$, and $R_F=R_G\cup\{a\}$. The matching consists of the original pairs between $S_0$ and $K_0$, the pairs between $A'$ and $B'$, and the leg $vu$.
	
	Since $K_F=K_G$, the clique edges of $F$ already belong to $G$. Every vertex of $R_G$ is already adjacent to every vertex of $K_G$. The only additional edges required, apart from $uv$, are the edges joining $a$ to $K_G\setminus\{u\}=K_0\cup B'$. Thus $F$ requires $1+q+(t-1)$ fill edges.
	
	We now obtain a lower bound on the number of fill edges in $H$. The vertex $a$ belongs to the $R$-set of $H$. Since its only clique neighbor in $G[C_u]$ is $u$, the $q$ edges joining $a$ to $K_0$ are fill edges of $H$. Furthermore, every vertex $x\in A'$ belongs to the $R$-set of $H$. Since $x\in S_G$ in $G[C_u]$, it is nonadjacent in $G$ to every vertex of $K_0\cup\{u\}$. Hence, for each $x\in A'$, the $q+1$ edges joining $x$ to $K_0\cup\{u\}$ are fill edges of $H$. Therefore $N
	\ge
	1+q+(q+1)(t-1)$. Comparing this with the number of fill edges in $F$, we obtain
	\[
	\begin{split}
	N-|E(F)\setminus E(G)|
	&\ge
	q+(q+1)(t-1)-q-(t-1)\\
	&=
	q(t-1)>0,
	\end{split}
	\]
	because $q\ge2$ and $t>1$.
	This contradicts the optimality of $H$.
	
	Thus $t=1$. It follows that $A=\{a\}$, $B=\{u\}$, and hence $S_G=S_0\cup\{a\}
	=
	(S\setminus\{v\})\cup\{a\}$, while $K_G=K_0\cup\{u\}=K$. In the canonical completion, the only fill edges are $uv$ and the
	edges joining $a$ to $K_G\setminus\{u\}$. There are $|K_G|-1$ such edges. Therefore $N
	=
	1+(|K_G|-1)
	=
	|K_G|$. This proves part (ii)(c).
	
	The three possibilities $u\in R_G$, $u\in S_G$, $u\in K_G$ are mutually exclusive and exhaustive. Hence parts (ii)(a)--(c)
	complete the case $|K|\ge3$.
	
	The case in which $u\in S$, $v\in K$, and $f(u)=v$ follows
	symmetrically by interchanging the roles of $u$ and $v$.
\end{proof}

%\subsection
\medskip\noindent
\textbf{Case 3: The root node of the $P_4$-sparse tree~$T_H$ of the solution~$H$ is a $2$-node corresponding to a thick spider~$(S,K,R)$.}
\ According to our convention, $|S| = |K| \ge 3$.

\smallskip\noindent
\emph{Case 3a. The vertices $u,v$ belong to $R$.}
We first consider the case that $S \cup K \subset C_u$ or $S \cup K \subset C_v$; in a fashion similar to the proof of Lemma~\ref{lemma:2a_1}, we can show:

\begin{lemma} \label{lemma:3a_1}
	Suppose that an optimal solution~$H$ of the ($P_4$-sparse-2CC,{$+$}$1$)-Min\-Edge\-Addition Problem for a $P_4$-sparse graph~$G$ and an added non-edge~$uv$ is a thick spider~($S,K,R$) with $u,v \in R$.
	If $S \cup K \subseteq C_u$ then $G[C_u]$ is a thick spider~$(S_G,K_G,R_G)$ and $S = S_G$ and $K = K_G$, i.e., $T_{u,1}(H) = T_{u,1}(G[C_u])$.
	\\
	A symmetric result holds if $S \cup K \subseteq C_v$.
\end{lemma}

\begin{proof}

    Case A, where the root of $T_{G[C_u]}$ is a 1-node, is impossible by the same argument as in Lemma~6(A). 
    
    Suppose next that the root corresponds to a thin spider $(S_G,K_G,R_G)$. If some $w\in K_G\setminus K$, then the adjacencies of $w$ imply $S\subseteq S_G$. Moreover, $N_G(S)\subseteq K$ and $|N_G(S)|=|S|=|K|$, so $N_G(S)=K$. Replacing $H[S\cup K]$ by $G[S\cup K]$ therefore yields a $P_4$-sparse completion with fewer fill edges, a contradiction. Hence $K_G\subseteq K$, and therefore $S_G\subseteq S$. Replacing $H[S_G\cup K_G]$ by $G[S_G\cup K_G]$ again yields a completion with fewer fill edges. Thus this case is impossible. 
    
    Finally, let $G[C_u]$ be a thick spider $(S_G,K_G,R_G)$. For every $w\in K_G$, $|N_H(w)\cap C_u|\ge |C_u|-2$, whereas every $z\notin K$ satisfies $|N_H(z)\cap C_u|\le |C_u|-|S|\le |C_u|-3$. 
    Hence $K_G\subseteq K$. For every $p\in K_G$, let $s_p\in S$ be its unique non-neighbor in $H$. Since $S\subseteq C_u$ and $G$ is a spanning subgraph of $H$, the non-edge $ps_p$ of $H$ is also a non-edge of $G[C_u]$. But $p$ has exactly one non-neighbor in $G[C_u]$, namely its matched vertex in $S_G$. Thus $s_p\in S_G$. As $p$ ranges over $K_G$, these vertices are distinct, so $S_G\subseteq S$, and the non-neighbors in $S$ of the vertices of $K_G$ are precisely the vertices of $S_G$.
    
    Let $K_2=K\setminus K_G$. Then $K_2\subseteq R_G$. Since the non-neighbors in $S$ of the vertices of $K_G$ are precisely the vertices of $S_G$, the non-neighbor in $S$ of every vertex of $K_2$ belongs to $S\setminus S_G$. Hence every vertex of $K_2$ is adjacent in $H$ to every vertex of $S_G$. On the other hand, since $K_2\subseteq R_G$, no vertex of $K_2$ is adjacent in $G$ to a vertex of $S_G$. Thus all $|K_2||S_G|$ such edges are fill edges.

    If $K_2\ne\emptyset$, form a thick spider with stable set $S_G$, clique $K_G$, and $R$-part $V(G)\setminus(S_G\cup K_G)$, keeping $H[V(G)\setminus(S_G\cup K_G)]$ unchanged. This is a $P_4$-sparse completion of $G+uv$ and omits those $|K_2||S_G|$ fill edges, contradicting the optimality of $H$. Thus $K_2=\emptyset$. Since $|S|=|K|$ and $|S_G|=|K_G|$, we obtain $S=S_G$ and $K=K_G$, and hence $T_{u,1}(H)=T_{u,1}(G[C_u])$. 
\end{proof}

However, unlike Case~2a,
%it turns out that 
this is the only possibility in this case.

\begin{lemma} \label{lemma:3a_2}
	Suppose that an optimal solution~$H$ of the ($P_4$-sparse-2CC,{$+$}$1$)-Min\-Edge\-Addition Problem for a $P_4$-sparse graph~$G$ and a non-edge~$uv$ is a thick spider~($S,K,R$) with $u,v \in R$. Then, it is not possible that $S \cup K$ contains vertices from both $C_u$ and $C_v$.
\end{lemma}
\begin{proof}
	Suppose for contradiction that $S \cup K$ contains a vertex in $C_u$ and a vertex in $C_v$. Because $H[C_u]$ and $H[C_v]$ are connected (Observation~\ref{obs: solution_prop}(ii)), there exist vertices $u' \in K \cap C_u$ and $v' \in K \cap C_v$ and let $u'', v''$ be the non-neighbors in $S$ of $u', v'$ respectively. Then, if $u'' \in C_u$, $u'$ is incident on $|(S \cup K \cup R) \cap C_v| = |C_v|$ fill edges in $H$ whereas if $u'' \in C_v$, $u'$ is incident on $|(S \cup K \cup R) \cap C_v| - 1 = |C_v| - 1$ fill edges; a symmetric result holds for $v'$ and $v''$. Before proceeding, we note that by making $u$ universal in $G$, we would need at most $|V(G)| - 2$ fill edges since $deg_G (u) \ge 1$ because $|C_u| \ge 2$ and $G[C_u]$ is connected (Observation~\ref{obs: solution_prop}(i)). Next, we distinguish the following cases:
    If $u''\in C_u$ and $v''\in C_v$, counting the fill edges incident with $u'$ and $v'$ across the two components, together with $uv$, gives $N\ge |V(G)|$. If $u''$ and $v''$ lie in the same component, the analogous count gives $N\ge |V(G)|-1$. 
    
    In the remaining case $u''\in C_v$ and $v''\in C_u$, the edges incident with $u',v'$ contribute $|V(G)|-3$ fill edges besides $uv$, and every vertex of $K\setminus\{u',v'\}$ contributes at least one further fill edge. Hence $N\ge1+(|V(G)|-3)+(|K|-2)\ge |V(G)|-1$. Thus in every case $N\ge |V(G)|-1$, whereas making $u$ universal requires at most $|V(G)|-2$ fill edges, a contradiction. 
	% \begin{itemize}
	% 	\item $u'' \in C_u$ and $v'' \in C_v$:
	% 	Then, in $H$, the number~$N$ of fill edges is $N \ge |C_v| + |C_u| - 1 + 1$
	% 	where we subtract $1$ for the double counted fill edge~$u' v'$ and we add $1$ for the added non-edge~$uv$, which implies that $N \ge |V(G)|$ in contradiction to the optimality of the solution~$H$.
	% 	%
	% 	\item $u'', v'' \in C_u$ or $u'', v'' \in C_ v$:
	% 	In either case, as in the previous item, in $H$, the number~$N$ of fill edges is $N \ge \left( |C_u| + |C_v| - 1 \right) - 1 + 1$, which implies that $N \ge |V(G)| - 1$, again a contradiction to the optimality of $H$.
	% 	%
	% 	\item $u'' \in C_v$ and $v'' \in C_u$:
	% 	Then, in $H$, in addition to the fill edge~$u v$ and the $(|C_u| - 1) + (|C_v| - 1) - 1 = |V(G)| - 3$ fill edges incident on $u', v'$, we note that any vertex in $K \setminus \{u',v'\}$ is adjacent to both $u'', v''$, thus being incident to at least $1$ fill edge, for a total of at least $1 + (|V(G)| - 3) + (|K| - 2) = |V(G)| + |K| - 4 \ge |V(G)| - 1$ fill edges, again a contradiction to the optimality of $H$.
	% \end{itemize}
\end{proof}

\smallskip\noindent
\emph{Cases 3b and 3c. At least one of the vertices $u,v$ belongs to $S \cup K$.}

\begin{lemma} \label{lemma:3bc}
	If there exists an optimal solution~$H$ of the ($P_4$-sparse-2CC,+$1$)-Min\-Edge\-Addition Problem for the union of $G[C_u]$ and $G[C_v]$ and the non-edge $uv$ such that the root node of the $P_4$-sparse tree corresponding to $H$ is a $2$-node corresponding to a thick spider~$(S,K,R)$ with at least one of $u,v$ in $S \cup K$, then there exists an optimal solution of the same problem which results from making $u$ or $v$ universal in $G$.
\end{lemma}
\begin{proof}
	First, note that a vertex in the set~$K$ needs exactly $1$ additional fill edge to become universal in $H$. The idea of the proof is to show that in each case at least one of $u,v$ belongs to $K$ and that by making it universal in $H$, we get an optimal solution that is no worse than $H$. Furthermore, recall that we consider that in a thick spider $|K| = |S| \ge 3$.

\smallskip\noindent
\emph{Case~3b: one of the vertices $u,v$ belongs to $R$ and the other belongs to $S \cup K$, which implies that in fact it belongs to $K$.}
\ Without loss of generality, we assume that $u \in R$ and $v \in K$. We show that $|C_v| \ge 2$. Otherwise, $C_v = \{v\}$, and we could get a solution~$H'$ with fewer fill edges than $H$ by removing $v$ and all incident edges from $H$ (the resulting graph is still $P_4$-sparse) and by adding fill edges incident on $u$ to all its non-neighbors including $v$, a contradiction; note that in $H$, $v$ is incident on $|V(G)| - 2$ fill edges (including $uv$) whereas $u$ is adjacent to at least $|K|-1 \ge 2$ vertices other than $v$ which implies that it has at most $|V(G)| - 3$ non-neighbors including $v$. Thus, $|C_v| \ge 2$.
	Additionally, it holds that $K \cap C_v = \{v\}$ since otherwise, in addition to the fill edges incident on $v$ in $H$, we would have at least $2$ more fill edges whereas by making $v$ universal in $H$, we get a solution that requires fewer fill edges than $H$, a contradiction; to see this, note that if $|K \cap C_v| \ge 3$ there exist at least $2$ more fill edges connecting $u$ to each of the vertices in $(K \cap C_v) \setminus \{v\}$, whereas if $|K \cap C_v| = 2$ there exist at least $2$ more fill edges connecting the vertex in $(K \cap C_v) \setminus \{v\}$ to $u$ and to the vertices in $K \cap C_u$ where $|K \cap C_u| = | K \setminus C_v| \ge 1$.

	Therefore, $|C_v| \ge 2$ and $K \cap C_v = \{v\}$. In fact, $(C_v \setminus \{v\}) \subseteq S$; if there existed a vertex in $C_v \cap R$ then, in addition to the fill edges incident on $v$ in $H$, we would have at least $2$ more fill edges connecting that vertex to the vertices in $K \cap C_u$, again implying that making $v$ universal in $H$ would lead to a solution with fewer fill edges than $H$, a contradiction. Since $K \cap C_v = \{v\}$, each vertex in $S$ is adjacent to at least $1$ vertex in $K \cap C_u$ and thus the optimality of~$H$ (versus the solution with $v$ being universal in $H$) implies that $|C_v| = 2$, $|K| = 3$, and the only fill edges are those between $C_u$ and $C_v$, for a total of $|R|+4$ fill edges (including $uv$), matching the solution obtained by making $v$ universal in $G$.

%  \begin{figure}[t!]
% 		%\hrule\medskip\medskip	
% 		\centering
% 		\includegraphics{thick3b.eps}
% 		%\smallskip\medskip\hrule\medskip
% 		\caption{The number of fill edges (in green) is $|R| + 4$ including $uv$ (in red).}
% 		\label{fig:thick3b}
% \end{figure}

\smallskip\noindent
\emph{Case~3c: the vertices $u,v$ belong to $S \cup K$.} Then, because the vertices in $S$ form an independent set in $H$, at least one of $u,v$ belongs to $K$; without loss of generality, let us assume that $u \in K$. We consider the following cases:

	\begin{itemize}
		\item[(a)] $K \subseteq C_u$: Because $H[C_v]$ is connected (Observation~1(ii)), then $C_v = \{v\}$ which implies that $v$ is incident on fill edges to $u \in K$ and to $|K|-2 \ge 1$ more vertices in $C_u$; then, the optimality of the solution~$H$ (versus the solution with $u$ being universal in $G$) implies that $|K| = 3$, and the fill edges are those connecting $v$ to the vertices in $K \subseteq C_u$ (a total of $2$ fill edges) matching the number of fill edges if $u$ is universal in $G$.
		\item[(b)] $K \cap C_v \ne \emptyset$: We show that $v \not\in K$.  Otherwise, let $w \in K - \{u,v\}$ and $w' \in S$ be the non-neighbor of $w$. If $w, w' \in C_u$, then, in addition to the fill edges incident on $u$ in $H$, $H$ contains the $2$ fill edges $v w$ and $v w'$, a contradiction to the optimality of $H$ compared to the solution with $u$ being universal in $G$; if $w, w' \in C_v$, the case is symmetric considering $v$ being universal in $G$. So consider that one of $w,w'$ belongs to $C_u$ and the other in $C_v$; due to symmetry, we can assume that $w \in C_u$ and $w' \in C_v$. Then $H$ contains the fill edges $v w$ and $u w'$. Now consider the non-neighbor~$x$ of $u$ in $S$, which is adjacent to both $v$ and $w$; if $x \in C_u$, then $H$ also contains the fill edge~$vx$ and thus is not optimal compared to the solution with $u$ being universal in $G$ whereas if $x \in C_v$, $H$ contains the fill edge~$wx$ and again $H$ is not optimal compared to the solution with $u$ being universal in $G$.
		
		Thus $v \not\in K$; since $u,v \in S \cup K$, then $v \in S$. Since $K \cap C_v \ne \emptyset$ and $H[C_v]$ is connected (Observation~\ref{obs: solution_prop}(i)), there exists $w \in K \cap C_v$ with $w$ being adjacent to $v$. Then we can show that $K \setminus \{u\} \subseteq C_v$; otherwise, there would exist a vertex $x \in (K \setminus \{u\}) \cap C_u$ and if $x' \in S$ is a common neighbor of $w,x$, the graph~$H$ would include the fill edges $w x$ and one of $w x'$ or $x x'$ (depending on whether $x'$ belongs to $C_u$ or to $C_v$, respectively) and thus $H$ is not optimal compared to the solution with $u$ being universal in $G$, a contradiction. In a similar fashion, $R \subseteq C_v$ for otherwise $H$ would contain the at least $2$ fill edges from any vertex in $R \cap C_u$ to all the vertices in $K \setminus \{u\}$ and would not be optimal compared to the solution with $u$ being universal in $G$. A similar argument proves that there is at most $1$ vertex in $S \cap C_u$ and that all the vertices in $S$ that are adjacent to at least $2$ vertices in $K \setminus \{u\}$ need also belong to $C_v$.
		
		Then, either (i)~$C_u = \{u\}$ or (ii)~$C_u = \{u, z\}$ (where $z \in S$ is a neighbor of $u$) and $|K| = 3$ (otherwise $z$ would be adjacent to at least $2$ vertices in $K \setminus \{u\}$ and thus would need to belong to $C_v$, a contradiction). In the former case, $H$ contains $|R| + 2\,|K| - 2$ fill edges incident on $u$, whereas in the latter, $|R| + 3\,|K| - 5$ fill edges incident on $u$ and $z$. However, we can show that in either case, we get a $P_4$-sparse graph by replacing these fill edges with fewer ones. In the following, let $u' \in S$ be the non-neighbor of $u$ and $v', z' \in K$ be the non-neighbors of $v,z$ respectively.

        \begin{figure}[t!]
            \centering
        
            \begin{minipage}{0.47\textwidth}
                \centering
                \includegraphics[width=\textwidth]{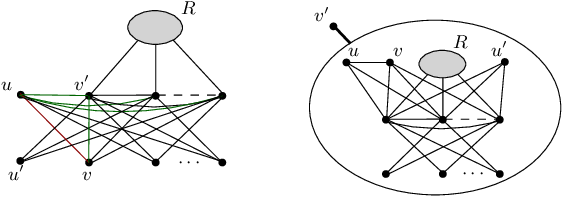}\\[-1mm]
                \textbf{(a)}
            \end{minipage}
            \hspace{0.04\textwidth}
            \begin{minipage}{0.47\textwidth}
                \centering
                \includegraphics[width=\textwidth]{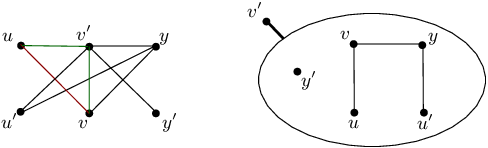}\\[-1mm]
                \textbf{(b)}
            \end{minipage}
        
            \bigskip\bigskip
        
            \begin{minipage}{0.47\textwidth}
                \centering
                \includegraphics[width=\textwidth]{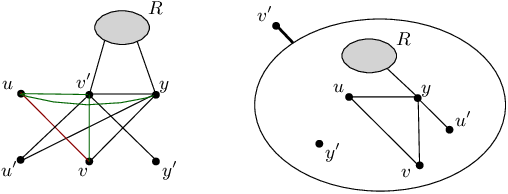}\\[-1mm]
                \textbf{(c)}
            \end{minipage}
            \hspace{0.04\textwidth}
            \begin{minipage}{0.47\textwidth}
                \centering
                \includegraphics[width=\textwidth]{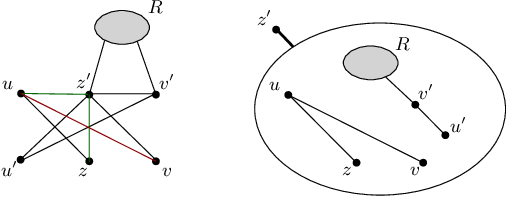}\\[-1mm]
                \textbf{(d)}
            \end{minipage}
        
            \caption{Constructions used in Lemma~12:
            (a) $C_u=\{u\}$, $|K|\ge4$;
            (b) $C_u=\{u\}$, $|K|=3$, $R=\emptyset$;
            (c) $C_u=\{u\}$, $|K|=3$, $R\ne\emptyset$;
            (d) $C_u=\{u,z\}$, $|K|=3$.
            Green edges are fill edges and $uv$ is red.}
            \label{fig:thick3c}
        \end{figure}

        % \begin{figure}[t!]
        % 	%\hrule\medskip\medskip	
        % 	\centering
        % 	\includegraphics[width=0.75\textwidth]{thick3ci_1.eps}
        % 	%\smallskip\medskip\hrule\medskip
        % 	\caption{(left)~The graph~$G$ with fill edges (in green) and the non-edge~$uv$ (in red) where $C_u = \{u\}$ and $|K| \ge 4$; (right)~the resulting graph after the addition of fill edges.}
        % 	\label{fig:thick3ci_1}
        % \end{figure}
          
        % \begin{figure}[t!]
        % 		%\hrule\medskip\medskip	
        % 		\centering
        % 		\includegraphics[width=0.75\textwidth]{thick3ci_2.eps}
        % 		%\smallskip\medskip\hrule\medskip
        % 		\caption{(left)~The graph $G$ with fill edges (in green) and the non-edge~$uv$ (in red) where $C_u = \{u\}$, $|K| = 4$, and $R= \emptyset$; (right)~the resulting graph after the addition of fill edges.}
        % 		\label{fig:thick3ci_2}
        % \end{figure}
         
        % \begin{figure}[t]
        % 			%\hrule\medskip\medskip	
        % 			\centering
        % 			\includegraphics[width=0.75\textwidth]{thick3ci_3.eps}
        % 			%\smallskip\medskip\hrule\medskip
        % 			\caption{(left)~The graph $G$ with fill edges (in green) and the non-edge~$uv$ (in red) where $C_u = \{u\}$, $|K| =3$, and $R \neq \emptyset$; (right)~the resulting graph after the addition of fill edges.}
        % 			\label{fig:thick3ci_3}
        % \end{figure} 

		\begin{itemize}
			\item[(i)] $C_u = \{u\}$.
			If $|K| \ge 4$, we use $|K| + 1$ fill edges ($|K|$ fill edges connecting $u$ to $v$ and to the vertices in $K \setminus \{u\}$ and $1$ more fill edge connecting $v$ to $v'$); $v'$ becomes universal in the resulting graph while the remaining vertices induce a thick spider with $S' = S \setminus \{u',v\}$, $K' = K \setminus \{u,v'\}$, and $R' = R \cup \{u, v, u'\}$ (Figure~\ref{fig:thick3c}(a)). If $|K| = 3$ and $R = \emptyset$, we use $3$ fill edges to connect $u$ to $v$ and to $v'$, and to connect $v$ to $v'$; in the resulting graph (Figure~\ref{fig:thick3c}(b)), $v'$ is universal and in the subgraph induced by the remaining vertices, the vertex in $S \setminus \{v, u'\}$ becomes isolated and the other vertices induce a $P_4$. If $|K| = 3$ and $R \ne \emptyset$, we use $4$ fill edges by additionally using the fill edge~$u y$ where $y$ is the vertex in $K \setminus \{u, v'\}$; in the resulting graph (Figure~\ref{fig:thick3c}(c)), the vertex~$v'$ and the vertex in $S \setminus \{v, u'\}$ are as in the case for $|K| = 3$ and $R = \emptyset$, vertex~$y$ is universal in the subgraph induced by the remaining vertices which in turn induce a disconnected graph with connected components $R$, $\{u'\}$, and $\{u,v\}$.
			\item[(ii)] $C_u = \{u, z\}$ and $|K| = 3$.
			In this case, we use $3 = |K|$ fill edges to connect $v$ to $u$ and to connect $u$ and $z$ to $z'$, and then $z'$ becomes universal in the resulting graph ($z'$ is universal in $H[C_v]$), in which the remaining vertices induce a disconnected subgraph with connected components $R \cup \{u',v'\}$ and $\{v,u,z\}$ (Figure~\ref{fig:thick3c}(d)).
		\end{itemize}
		In either case, we get a contradiction to the optimality of the solution~$H$ (note that for any $|K| \ge 4$ it holds that $|K|+1 < 2\,|K| - 2 \le |R| + 2\,|K| - 2$ whereas for $|K| = 3$ we have: $3 < 4 = 2\,|K| - 2 \le |R| + 2\,|K| - 2$; for $R \ne \emptyset$, $4 < |R| + 4 = |R| + 2\,|K| - 2$; lastly, $3 < 4 = 3\,|K| - 5 \le |R| + 3\,|K| - 5$).
	\end{itemize}
\end{proof}

\section{Adding a Non-Edge to a General $P_4$-sparse Graph}
\label{ch:general_algo}

It is not difficult to see that the following fact holds.

\begin{observation} \label{obs:algorithm}
Let $G$ be a $P_4$-sparse graph, $T_G$ be the $P_4$-sparse tree of $G$, $uv$ be a non-edge that we want to add, and let $t$ be the least common ancestor of the tree leaves corresponding to $u,v$ in $T_G$.
\begin{itemize}
\item If $t$ is a $0$-node and $C_u$ ($C_v$ resp.) are the connected components containing $u$ ($v$ resp.) in $G$ after we have removed all of their common neighbors, then an optimal solution of the ($P_4$-sparse,$+1$)-Min\-Edge\-Addition Problem for the graph~$G$ and the non-edge~$uv$ can be obtained from $G$ after we have replaced the induced subgraph~$G[C_u \cup C_v]$ by an  optimal solution of the ($P_4$-sparse-2CC,$+1$)-Min\-Ed\-ge\-Ad\-di\-tion Problem for the graph union $G[C_u] \cup G[C_v]$ and the non-edge~$uv$.
\item If $t$ is a $2$-node and $Q$ is the corresponding spider, then an optimal solution of the ($P_4$-sparse,$+1$)-Min\-Edge\-Addition Problem for the graph~$G$ and the non-edge~$uv$ can be obtained from $G$ after we have replaced the induced subgraph~$G[Q]$ by an  optimal solution of the ($P_4$-sparse,$+1$)-Min\-Ed\-ge\-Ad\-di\-tion Problem for the spider~$Q$ and the non-edge~$uv$.
\end{itemize}
\end{observation}

Our algorithm consists of two procedures: the Algorithm {\sc $(P_4$-sparse)-MEA} for solving the ($P_4$-sparse,$+1$)-Min\-Ed\-ge\-Ad\-di\-tion problem for a $P_4$-sparse graph~$G$ and a non-edge~$uv$ and the Algorithm {\sc $(P_4$-sparse-$2$CC)-MEA} for the problem on a $2$-component graph and a non-edge~$uv$ (as the recursion descends, the vertex set containing $u$ or $v$ may become disconnected; if its root is a $0$-node, we keep the child subtree containing the endpoint, without adding fill edges).

In light of the lemmas in Section~\ref{ch:inside_spider} and Observation~\ref{obs:algorithm}, Algorithm {\sc   $(P_4$-spar\-se)-MEA} for a given $P_4$-sparse graph~$G$ and a non-edge~$uv$ computes the $P_4$-sparse tree~$T_G$ of $G$ and the least common ancestor of the leaves corresponding to $u$ and $v$ in $T_G$: if it is a $2$-node, we apply the results in Lemmas~\ref{lemma:thin}-\ref{lemma:thick} and in the S-R case, we call Algorithm {\sc ($P_4$-sparse-$2$CC)-MEA} returning the sum of the returned value and the number of fill edges incident on $S \cup K$ given in Lemmas~\ref{lemma:thin}(iii) and \ref{lemma:thick}(iii); if it is a $0$-node, we apply Observation~\ref{obs:algorithm}, compute the connected components that include $u$ and $v$, and call Algorithm {\sc ($P_4$-sparse-$2$CC)-MEA} returning the returned number of fill edges. 
%Pseudocode of Algorithm {\sc $P_4$-sparse-Edge-Addition} is given below. 

Algorithm {\sc ($P_4$-sparse-$2$CC)-MEA} relies on the lemmas of Section~\ref{sec:cc}. In its initial call, $C_u$ and $C_v$ are the connected components containing $u$ and $v$, respectively; in recursive calls they denote the current vertex sets containing these vertices and need not induce connected graphs. Next, we give pseudocode for the algorithm. For a vertex set $C_x$ containing $x$ such that $G[C_x]$ is connected, let $L_{\rm leg}(C_x,x)$ denote the minimum number of fill edges given by Lemma~\ref{lemma:9c} for a leg-completion, and let $L_{\rm leg}(C_x,x)=+\infty$ if no such completion exists.

\medskip
\noindent
\textbf{Algorithm \textsc{($P_4$-sparse-$2$CC)-MEA}}
$(G,u,C_u,T_G(G[C_u]),v,C_v,T_G(G[C_v]))$

\smallskip
\small
\begin{algorithmic}

\If{$|C_u|=1$ and $|C_v|=1$}  \Return $1$ 
    %\State \Return $1$
\EndIf

\If{the root of $T_G(G[C_u])$ is a $0$-node}
    \State $A \gets C_u \setminus V(T_{u,1}(G[C_u]))$; $T \gets T_G(G[C_u])$ without $T_{u,1}(G[C_u])$
    %\State $T \gets T_G(G[C_u])$ without $T_{u,1}(G[C_u])$
    \State \Return {\sc ($P_4$-sparse-$2$CC)-MEA}
    $(G,u,A,T,v,C_v,T_G(G[C_v]))$
\EndIf

\If{the root of $T_G(G[C_v])$ is a $0$-node}
    \State $A \gets C_v \setminus V(T_{v,1}(G[C_v]))$; $T \gets T_G(G[C_v])$ without $T_{v,1}(G[C_v])$
    %\State $T \gets T_G(G[C_v])$ without $T_{v,1}(G[C_v])$
    \State \Return {\sc ($P_4$-sparse-$2$CC)-MEA}
    $(G,u,C_u,T_G(G[C_u]),v,A,T)$
\EndIf

\State $special \gets +\infty$
\ForAll{applicable special constructions in
Lemmas~\ref{lemma:1a_1}, \ref{lemma:2a_2},
\ref{lemma:2b}, and~\ref{lemma:2c}}
    \State compute their number $N$ of fill edges
    \State $special \gets \min\{special,N\}$
\EndFor

\State $temp1 \gets |C_u|+|C_v|-1-\deg_{G[C_u]}(u)$
\State $temp2 \gets |C_u|+|C_v|-1-\deg_{G[C_v]}(v)$
\State $temp3,temp4,leg1,leg2 \gets +\infty$

\If{$|C_u|\ge2$}
    \State $A \gets C_u\setminus V(T_{u,1}(G[C_u]))$; $T \gets T_G(G[C_u])$ without $T_{u,1}(G[C_u])$
    %\State $T \gets T_G(G[C_u])$ without $T_{u,1}(G[C_u])$
    \State $temp3 \gets$
    \textsc{($P_4$-sparse-$2$CC)-MEA}
    $(G,u,A,T,v,C_v,T_G(G[C_v]))$
    \If{the root of $T_G(G[C_u])$ is a $1$-node} 
       \State $temp3 \gets temp3+
       |C_v|\,|V(T_{u,1}(G[C_u]))|$
    \Else
       \State $temp3 \gets temp3+
       |C_v|\,|V(T_{u,1}(G[C_u]))|/2$
    \EndIf
\EndIf

\State compute $temp4$ symmetrically

\If{$|C_v|=1$} $leg1 \gets L_{\rm leg}(C_u,u)$
    %\State $leg1 \gets L_{\rm leg}(C_u,u)$
\EndIf
\If{$|C_u|=1$} $leg2 \gets L_{\rm leg}(C_v,v)$
    %\State $leg2 \gets L_{\rm leg}(C_v,v)$
\EndIf

\State \Return
$\min\{temp1,temp2,temp3,temp4,special,leg1,leg2\}$

\end{algorithmic}
\normalsize

The algorithms can also be augmented to return a minimum-cardinality set of fill edges (including \(uv\)). By the lemmas above, the candidate forms considered by Algorithm {\sc ($P_4$-sparse-$2$CC)-MEA} are exhaustive: universal-vertex solutions, recursive constructions associated with \(T_{u,1}\) or \(T_{v,1}\), the special constructions of Lemmas 4, 7, 8, and 9A, and the leg-completions characterized in Lemmas 9B–9C. Hence, the minimum returned by the algorithm is optimal.

% The algorithms can be easily augmented to return a minimum-cardinality set of fill edges (including $uv$). The lemmas above cover all possible forms of an optimal solution. In particular, an optimal solution can be obtained by making $u$ or $v$ universal, by one of the recursive constructions associated with $T_{u,1}$ or $T_{v,1}$, by one of the applicable special constructions described in Lemmas~\ref{lemma:1a_1}, \ref{lemma:2a_2}, \ref{lemma:2b}, and~\ref{lemma:2c}, or by a leg-completion as characterized in Lemmas~\ref{lemma:2cbb} and~\ref{lemma:9c}. Hence, Algorithm {\sc ($P_4$-sparse-$2$CC)-MEA} computes the number of fill edges for all candidate forms of an optimal solution and returns their minimum, and therefore it returns the minimum number of fill edges required.

%The algorithms can be easily augmented to return a minimum-cardinality set of fill edges (including $uv$). The lemmas above cover all cases; in particular, the leg case of Case~2c.2 is handled by Lemmas~\ref{lemma:2cbb} and~\ref{lemma:9c}. Hence, the correctness of the algorithm follows by considering all possible labels of the root of the decomposition tree of an optimal completion.%The correctness of the algorithms follows from the lemmas in Sections \ref{ch:inside_spider} and \ref{sec:cc}, and Observation~\ref{obs:algorithm}.

\medskip\noindent
\textbf{Time and space Complexity.}
Let \(G\) have \(n\) vertices and \(m\) edges. The \(P_4\)-sparse tree \(T_G\) can be constructed in \(O(n+m)\) time and has \(O(n)\) nodes and height. The values \(L_{\mathrm{leg}}\), which occur only when one of the two current vertex sets is a singleton (Lemma 9B), can all be computed in \(O(n^2)\) total time and \(O(n)\) space.
For Algorithm \((P_4\text{-sparse-2CC})\)-MEA, a state is determined by a pair of positions on the paths in \(T_G\) from the leaves corresponding to \(u\) and \(v\) to their least common ancestor; hence there are \(O(n^2)\) states. Apart from the precomputed leg values, each state is computed in constant time from previously computed states, yielding \(O(n^2)\) total time. Moreover, only a constant number of consecutive diagonals of the state matrix, together with the \(O(n)\) boundary entries required for Lemma 4, need to be stored. Thus, the space requirement is \(O(n)\).

\begin{theorem}
	Let $G$ be a $P_4$-sparse graph on $n$ vertices and let $u,v$ be two non-adjacent vertices of $G$. Then, for the $(P_4\text{-sparse},+1)$-MinEdgeAddition Problem for $G$ and the non-edge $uv$, the minimum number of fill edges needed, including $uv$, can be computed in $O(n^2)$ time using $O(n)$ space.
    %Let $G$ be a $P_4$-sparse graph on $n$ vertices and $u,v$ be two non-adjacent vertices of $G$. Then for the ($P_4$-sparse,$+1$)-Min\-Ed\-ge\-Ad\-di\-tion Problem for the graph~$G$ and the non-edge~$uv$, we can compute the minimum number of fill edges needed (including $uv$) in $O(n^2)$ time using $O(n)$ space.
\end{theorem}

% \section {Open Problems}
% It would be interesting to find algorithms for the (${\cal C}$,$+1$)-Min\-Ed\-ge\-Ad\-di\-tion Pro\-blem for other important subclasses~${\cal C}$ of  perfect graphs, such as, chordal, interval, permutation, and comparability.

%\newpage

\end{document}